\documentclass[reqno,11pt]{amsart}

\usepackage{todonotes}
\usepackage{comment}
\usepackage[english]{babel}
\usepackage{amsfonts, amsmath, amssymb, amsthm}

\usepackage{hyperref,xcolor}
\usepackage{mathtools}
\usepackage[all]{xy}

\usepackage{amsthm}
\newtheorem{definition}{Definition}
\newtheorem{lemma}[definition]{Lemma}
\newtheorem{claim}[definition]{Claim}
\newtheorem{proposition}[definition]{Proposition}

\newtheorem{theorem}[definition]{Theorem}
\newtheorem{corollary}[definition]{Corollary}
\newtheorem{conjecture}[definition]{Conjecture}
\theoremstyle{remark}

\newtheorem{remark}[definition]{Remark}

\def\Re{\operatorname{Re}}

\def\C{\mathbb{C}}
\def\R{\mathbb{R}}

\def\S{\mathbb{S}}

\newcommand{\dVol}{\mathrm{dVol}}

\newcommand{\Lie}{\mathcal{L}}

\newcommand{\A}{\mathbb{A}}
\newcommand{\B}{\mathbb{B}}
\newcommand{\LL}{\mathbb{L}}

\newcommand{\sdet}{\mathrm{sdet}}

\DeclareMathOperator{\tr}{tr}

\newcommand{\intl}{\int\limits}

\title[Zeta Function from Field Theory]
{Ruelle Zeta Function from Field Theory}

\author[C. Hadfield]{Charles Hadfield}
\address{IBM T.J. Watson Research Center, 1101 Kitchawan Rd, Yorktown Heights, NY 10598, USA}
\email{charles.hadfield@ibm.com}

\author[S. Kandel]{Santosh Kandel}
\address{Mathematics Institute, University of Freiburg, Freiburg, 79104, Germany}
\email{skandel1@alumni.nd.edu}

\author[M. Schiavina]{Michele Schiavina}
\address{ETH Z\"urich
\newline\indent 
Department of Mathematics, R\"amistrasse 101, 8092 Z\"urich, Switzerland, \&
\newline\indent 
Institute for theoretical physics, Wolfgang-Pauli strasse 27, 8093 Z\"urich, Switzerland}
\email{micschia@phys.ethz.ch}

\begin{document}

\begin{abstract}
We propose a field-theoretic interpretation of Ruelle zeta function, and show how it can be seen as the partition function for $BF$ theory when an unusual gauge fixing condition on contact manifolds is imposed. This suggests an alternative rephrasing of a conjecture due to Fried on the equivalence between Ruelle zeta function and analytic torsion, in terms of homotopies of Lagrangian submanifolds.
\end{abstract}

\maketitle


\tableofcontents

\section*{Introduction}
Quantum field theory is a useful tool in many areas of pure and applied mathematics. It provides a number of precise answers, often involving insight coming from statements that are theorems in finite dimensions, and that need to be appropriately checked and generalised in infinite dimensions.

A positive example of this is the interpretation by Schwarz of the Ray--Singer analytic torsion in terms of a partition function for a degenerate functional  
\cite{schwarz1978partition, schwarz1979partition}. 
The main ingredient in Schwarz's construction is a  {topological field theory} involving differential forms, which enjoys a symmetry given by the shift of closed forms by exact ones
\cite{BBRT91,cattaneo1995topological, cattaneo2001higher}. This is known nowadays with the name of $BF$ theory.

From a field-theoretic point of view, such symmetry needs to be removed, or  {gauge fixed}, as it represents a fundamental redundancy in the description. One possible way to do this is by choosing a reference metric $g$ and enforcing a $g$-dependent condition on fields\footnote{This is often called  {Lorenz gauge fixing}.}. It allows to compute the partition function of the theory - the starting point for quantum considerations on the system - and one is left to show that the choice of metric is immaterial. The proof that such choice of metric is irrelevant was given by Schwarz for the partition function of abelian $BF$ theory, and it is tantamout to the statement of independence of the analytic torsion on the metric used to define a Laplacian on the underlying manifold.

There are several ways of encoding a choice of gauge fixing within the framework of field theory, starting from the original idea of Faddeev and Popov 
\cite{faddeev2016feynman}, 
later understood in terms of Lie algebra cohomology by Becchi, Rouet, Stora and Tyutin 
\cite{becchi1974abelian, becchi1975renormalization, becchi1976renormalization, tyutin1975gauge}. 
A more general approach follows the ideas of by Batalin and Vilkovisky \cite{batalin1983quantization, batalin1984gauge}, 
and implements the choice of a gauge as the choice of a Lagrangian submanifold in an appropriate (graded)-symplectic manifold of fields $\mathcal{F}$. In this context, gauge-fixing independence is phrased in terms of isotopies of embedded Lagrangian submanifolds, and needs to be proven in some appropriate regularisations scheme. In finite dimensions this is a theorem: the partition function for an action functional $S$ that satisfies the  {quantum master equation} (a differential condition on $S$) does not depend on the choice of a particular Lagrangian submanifold inside a smooth family $\LL_t\subset\mathcal{F}$. Observe that this statement can be phrased as local constancy of the partition function w.r.t a parametrisation of the Lagrangian homotopy.

Leaving field theory aside for a moment, consider a closed manifold $M$ endowed with an Anosov vector field
(Definition~\ref{def:Anosov}).
The flow associated with the vector field is a typical example of a dynamical system displaying ``hard chaos"
\cite{gutzwiller1991chaos}.
An important example of such a dynamical system is obtained from a Riemannian manifold $(\Sigma, g)$ whose sectional curvature is negative; then $M=S^*_g\Sigma$, the unit-cotangent bundle of $\Sigma$ (a sphere bundle), is such that the Reeb vector field $X$ associated with the natural contact structure is an Anosov vector field and its flow coincides with the geodesic flow.

If an Anosov flow (generated by the vector field $X$) admits closed orbits, one defines a (dynamical) Ruelle zeta function $\zeta_X(\lambda)$ to count lengths of closed orbits associated with the flow in a similar spirit to how the Riemann zeta function counts prime numbers
\cite{ruelle1976zeta, ruelle1986resonances}. 
The zeta function also may be defined in the presence of a representation $\rho$ of $\pi_1(M)$ which provides a flat vector bundle over $M$, this leads to a  {twisted} zeta function. 
The chaotic nature of the dynamical system ensures that the zeta function is well-defined for $\Re(\lambda)\gg 1$, however work is required to show that the function extends meromorphically to the whole complex plane
\cite{giulietti2013anosov}.
Conjecture~\ref{conj:fried}, due to Fried 
\cite{fried1986analytic},
proposes that when $M=S^*_g\Sigma$, the Ruelle zeta function (evaluated at zero) exactly computes the analytic torsion of the associated sphere bundle. To connect this to field theory, we observe that this means that Ruelle zeta function is expected to compute - in Schwarz's terms - the partition function of $BF$ theory in a given (metric dependent) gauge fixing.

Fried's conjecture has received considerable attention recently. 
The proposed equality was confirmed in
\cite{fried1986analytic} 
for $\Sigma$ a hyperbolic manifold, and conjectured in 
\cite{fried1987lefschetz} 
that it also holds for compact locally symmetric spaces with non positive curvature. 
Conjecture~\ref{conj:fried}, as we state it, appears in 
\cite{fried1995meromorphic}.
A more precise version for locally symmetric manifolds has been proved in
\cite{shen2017analytic}
following
\cite{moscovici1991r}.
In the variable curvature case a perturbative result has been obtained in
\cite{dang2018fried}
and extended in
\cite{chaubet2019dynamical}.
A surprising result in the case of surfaces with variable negative curvature, but without reference to an acyclic representation, showed the zeta function at zero is determined by the topology of the surface
\cite{dyatlov2017ruelle}.
This has been extended to the case of surfaces with boundary
\cite{hadfield2018zeta}
and to higher dimensional closed manifolds perturbatively close to hyperbolic space
\cite{kuster2019pollicott}.
However Fried's conjecture, along with its three star bounty
\cite[Section 3, footnote 6]{zworski2017mathematical},
remains open.

\subsection*{What to expect from this paper}
We present a new class of gauge fixings for $BF$ theory on contact manifolds based on the Reeb vector field associated with the contact structure; we call this the  {contact gauge}
in Definition~\ref{def:contact-gauge}.
We then go on to show that, on sphere bundles with an Anosov--Reeb vector field, the Ruelle zeta function can be interpreted as an appropriately regularised determinant for the Lie derivative operator $\mathcal{L}_X$ on $k$-forms in the kernel of the contraction $\iota_X$. Taking this regularised determinant as the definition for the partition function of $BF$ theory in the contact gauge allows us to conclude that this coincides with the Ruelle zeta function. This point of view is analogous to Schwarz's calculation of the partition function of $BF$ theory (in the metric gauge), whose output is the analytic torsion, and to the recent proof of Chern--Gauss--Bonnet Theorem that has been given with similar techniques in \cite{Berwick2015} (see also \cite{cordes1995lectures}).

As a consequence, we relate the  {expected} gauge-fixing independence of the partition function of $BF$ theory to Fried's conjecture; in particular, we show how modern proofs of the conjecture for certain classes of manifolds (see \cite{dang2018fried}) can be taken as a proof of gauge-fixing independence. On the other hand, we believe that the field theoretic presentation of the Ruelle zeta function provided in this paper will allow the problem to be tackled from a different angle: by means of homotopies of Lagrangian submanifolds.

To this aim, we setup a convenient construction to compare Anosov vector fields that are related to a choice of a metric on a base manifold $\Sigma$ (Section \ref{s:homotopies}). By means of a natural construction for sphere bundles, we map smooth paths of metrics into smooth paths of Anosov vector fields, effectively constructing an isotopy between their associated Lagrangian submanifolds. This, together with the crucial local-constancy results of \cite{dang2018fried}, allow us to test our approach to the known case of 2d surfaces --- Theorem \ref{thm:Friedconjecture2d} provides an alternative proof of Fried's conjecture on surfaces --- and interprets it as gauge-fixing independence for $BF$ theory. 

From the point of view of algebraic topology, this result suggests that, under certain assumptions, $\iota_X$ can be made into a chain contraction for the de Rham complex, namely one can construct $\eta_X=(\mathcal{L}_X)^{-1} \iota_X$ with the appropriate conditions of nondegeneracy of $\mathcal{L}_X$. This intepretation appears to be related with the notion of a  {dynamical torsion} introduced in \cite{chaubet2019dynamical}. There appears to be a sweet spot at the intersection of Anosov and Reeb vector fields where the independence of the ``torsion'' of the de Rham complex on the choice of a chain contraction, and independence of the partition function of $BF$ theory on a choice of gauge fixing appear to be aspects of the same statement, expressed by Fried's conjecture.

This work is mostly addressed to the mathematical physics community working with or closely related to field theory in the Batalin--Vilkovisky formalism, but it is also aimed at the community interested in the microlocal analysis of Anosov/geodesic flows and Fried's conjecture. Therefore, we will present some basic background on field theory with symmetries to set the stage, terminology and expectations, but we will not present a complete treatment of the mathematics behind it. Results and constructions that will be somewhat assumed in this exposition of field theory can be found, e.g., in
\cite{anderson1992introduction,deligne1999classical, cattaneo2012classical,delgado2017lagrangian}.

Our main goal is to present a novel link between field theory and geometric and microlocal analysis, that will hopefully allow to import techniques across research fields, and stimulate fruitful interaction between scientific communities.

Section \ref{s:LFT} is an overview on Lagrangian field theory aimed at introducing the Batalin--Vilkovisky formalism and the problem of gauge fixing. It sets the stage for the field-theoretic interpretations that will follow.

Section \ref{Sect:geoprel} establishes the geometric conventions and notations, claryfing what incremental/alternative data one needs at different stages, and briefly describes the analytic torsion and Anosov dynamics.

In Section \ref{s:RZF} we introduce Ruelle zeta function and its $k$-form decomposition, and state Fried's conjecture. We interpret the zeta function as a regularised (super)determinant.

In Section \ref{s:BF} we describe a field theory called  {$BF$ theory}, we summarise the famous interpretation (due to Schwarz) of the analytic torsion in terms of the partition function of $BF$ theory, and introduce a new gauge fixing condition on contact manifolds. We show how, with that gauge-fixing, the partition function of $BF$ theory computes the Ruelle zeta function of the associated geodesic/Anosov flow.

Finally, in Section \ref{s:homotopies} we interpret Fried's conjecture in terms of gauge-fixing independence of $BF$ theory in the BV formalism, and suggest a  construction for sphere bundles that allows to present explicit homotopies between Lagrangian submanifolds.

\section{Lagrangian field theory, the Batalin--Vilkovisky formalism and regularised determinants}\label{s:LFT}
In this section we will review the basics of the Batalin--Vilkovisky (BV) formalism 
\cite{batalin1983quantization, batalin1984gauge} 
for Lagrangian field theories and how it handles gauge fixing. We will use a particular kind of regularisation based on the notion of flat traces to define determinants of operators and partition functions of quadratic functionals.

\subsection{Classical field theory, symmetries and quantisation}
The standard framework for Lagrangian field theories is as follows. To a compact manifold $M$, possibly endowed with extra geometric data, like a Riemannian metric or a contact structure, we associate a space of classical fields ${F}_M$, which is usually modelled on the space of sections of some vector bundle\footnote{More generally, a sheaf.} $E\to M$, together with a local functional $S_M$, called action functional. Local here means that it has the form of an integral over $M$ of a density-valued functional of the fields and a finite number of jets\footnote{There is an equivalent formulation of this in the variational bi-complex 
\cite{anderson1992introduction, deligne1999classical}, 
where the full jet bundle is taken into account.}:
\begin{equation}
    S_M= \int\limits_{M} L_M[\phi,\partial^I\phi],
\end{equation}
where $I$ is a finite multi-index and $L_M$ is called Lagrangian density. For simplicity we will consider compact manifolds without boundary, although it is possible to adapt the construction to non-compact ones or manifolds with boundary (see e.g. \cite{fredenhagen2013batalin, cattaneo2012classical}). 

The dynamical content of the theory is encoded in the Euler--Lagrange locus ${EL}[S_M]$, the space of solution of the Euler--Lagrange equations coming from the variational problem for $S_M$. In other words, the Euler--Lagrange locus is the set of critical points of $S_M$.

The action functional might enjoy a symmetry. That is, it might be invariant under some transformation of the fields, for example when considering Lie algebra actions on fields taking values in Lie algebra modules. Symmetries are usually described by a (smooth) distribution $D_M\subset TF_M$, and they make the critical points of the action functional degenerate\footnote{We will only be concerned with continuous symmetries.} and will become an issue when dealing with perturbative quantisation of the theory (see below). In what follows we will only consider symmetry distributions that are involutive.

Quantisation, loosely speaking, is meant to replace the (commutative) algebra of functions over the space of physical configurations of the system with some (noncommutative) algebra of operators over a suitable vector space, also called the space of quantum states. Without delving too much into how this is achieved in general, for our purposes it will be important to mention that one possible procedure starts by making sense of the following expression:
\begin{align}\label{Partitionfunction}
    Z=\int\limits \mathrm{exp}(\frac{i}{\hbar} S_M),
\end{align}
usually called the partition function, where the integral sign should ideally represent actual integration over ${F}_M$, with some measure. However, an appropriate integration theory for such (infinite dimensional) spaces of fields is generally not available, and one defines the previous expression as a formal power-series expansion in the parameter $\hbar$. This approach, however, requires the critical points of $S_M$ to be isolated, as it involves a saddle point or stationary phase approximation around critical points. It therefore automatically fails in the presence of symmetries, unless appropriate prescriptions are enforced. 

We choose to deal with this problem by means of the Batalin--Vilkovisky (BV) formalism.

\subsection{Cohomological approach and the BV complex}\label{Sect:BV}
Degenerate functionals are usually accompanied by involutive (symmetry) distributions $D_M$, and the space of inequivalent field configurations is the quotient ${EL}[S_M]/D_M$. Most of the times the quotient is singular and one looks for a replacement for it.

A  {resolution} of ${EL}[S_M]/D_M$ is given by a complex $(C^\bullet,d_C)$ such that\footnote{In some practical cases the vanishing of the negative cohomology is not guaranteed. We will anyway not need this condition in what follows.} for all $i>0$
\begin{align}\label{KTres}
    H^{-i}(C^\bullet)=0, \qquad H^0(C^\bullet)\simeq C^\infty({EL}[S_M]/D_M).
\end{align}

One way to obtain a resolution is by first localising to the submanifold ${EL}[S_M]$ constructing the Koszul--Tate complex, and then following the Chevalley--Eilenberg procedure to describe $D_M$-invariant functions on it (see \cite{stasheff1998secret} for a ``geometric" jet-bundle explanation of this and \cite{fredenhagen2013batalin} for a more ``algebraic" one).

The Batalin-Vilkovisky formalism is essentially the interpretation of said complex as the space of functions over a $(-1)$-symplectic graded manifold $(\mathcal{F}_{BV},\Omega_{BV})$ \cite{henneaux1990lectures, stasheff1998secret, cattaneo2011introduction, cattaneo2012classical, cattaneo2018perturbative},
whose degree-0 part coincides with the original space of fields $F_M$, endowed with an odd vector field of degree-$1$ $Q\in C^\infty(\mathcal{F}_{BV},T[1]\mathcal{F}_{BV})$ such that\footnote{$Q$ is essentially the derivation $d_C$ interpreted as a vector field.} $[Q,Q]=0$ and a degree-0 functional $\mathcal{S}_M\colon \mathcal{F}_{BV}\to \mathbb{R}$ satisfying
\begin{equation}\label{Eq:CME}
    \iota_Q\iota_Q\Omega_{BV}=\{\mathcal{S}_M,\mathcal{S}_M\}_{\Omega_{BV}}=0,
\end{equation}
with $\{\cdot,\cdot\}_{\Omega_{BV}}$ the Poisson bracket associated with the symplectic structure $\Omega_{BV}$, and the compatibility condition
\begin{equation}\label{HamiltonEquation}
    \iota_Q\Omega_{BV} = d \mathcal{S}_M.
\end{equation}
\begin{remark}
In infinite dimensions one models $\mathcal{F}_{BV}$ on some appropriate space of sections of the jet bundle of a vector bundle $E\to M$. Then, the de Rham differential in Equation~\eqref{HamiltonEquation} is replaced with $\delta$, the variation operator interpreted as the vertical differential on local functionals over $\mathcal{F}_{BV}$ (see \cite{anderson1992introduction}). Observe that we could take Equation \eqref{Eq:CME} as a definition of Poisson brackets in infinite dimensions.
\end{remark}

On $C^\infty(\mathcal{F}_{BV})$ one constructs another (second order) differential ($\Delta_{BV}^2=0$) called BV-Laplacian and defines  {gauge fixing} to be the choice of a Lagrangian submanifold $\mathbb{L}\subset \mathcal{F}_{BV}$. The main results are as follows:

\begin{theorem}\label{BVtheorem}
Let $(\mathcal{F}_{BV},\Omega_{BV})$ be a finite dimensional $(-1)$-symplectic graded manifold, with a measure $\mu$ and the BV Laplacian $\Delta_{\mu}$, a coboundary operator defined on $C^\infty(\mathcal{F}_{BV})$ such that for all $ f\in C^\infty(\mathcal{F}_{BV})$ we have
\begin{equation}\label{BV-div}
	\Delta_{\mu} f = -\frac12 div_{\mu}(X_f)
\end{equation}
with $X_f$ the Hamiltonian vector field of $f$ with respect to $\Omega_{BV}$. Assuming that $\Delta_{\mu_{BV}} f=0$ and $g=\Delta_{\mu_{BV}} h$, with $f,g,h\in C^\infty(\mathcal{F}_{BV})$ then:
\begin{itemize}
    \item for any Lagrangian submanifold $\mathbb{L}\subset\mathcal{F}_{BV}$
        \begin{align}
            \intl_{\mathbb{L}} g\ \mu_{BV}\vert_{\mathbb{L}} =0;
        \end{align}
    \item given a continuous family of Lagrangian submanifolds $\mathbb{L}_t$
        \begin{align}
            \frac{d}{dt}\intl_{\mathbb{L}_t} f\ \mu_{BV} \vert_{\mathbb{L}_t}= 0.
        \end{align}
\end{itemize}
\end{theorem}

\begin{remark}
The definition of the BV Laplacian in \eqref{BV-div} implies the relations 
\begin{align}
    \Delta_{\mu}(fg) = (\Delta_{\mu}f) g  + (-1)^{|f|} f (\Delta_{\mu}g) + (-1)^{|f|}\{f,g\}_{\Omega_{BV}}\\
    \Delta_{\mu} \{f,g\}_{\Omega_{BV}} = \{\Delta_{\mu}f,g\}_{\Omega_{BV}} 
        + (-1)^{|f|+1} \{f,\Delta_{\mu}g\}_{\Omega_{BV}},
\end{align}
making the tuple $\left(C^\infty(\mathcal{F}_{BV}),\, \cdot\, , \{\cdot,\cdot\}_{\Omega_{BV}},\Delta_\mu\right)$ into a BV algebra \cite{CFL}.
\end{remark}

\begin{remark}
Theorem~\ref{BVtheorem} is stated for finite dimensional manifolds. In this case $\Delta_{\mu}$ always exists. On infinite dimensional manifolds a number of complications arise. One needs an appropriate regularisation of $\Delta_{\mu}$, and the corresponding adaptations of statements in Theorem~\ref{BVtheorem} must be checked. In this paper we are interested in abelian $BF$ theory, which is a non-interacting topological field theory whose partition function (cf. Equation~\eqref{Partitionfunction}) is expected to be independent of the gauge fixing, and has been computed to be the Reidermeister (or equivalently) Analytic torsion of the de Rham complex \cite{schwarz1978partition, schwarz1979partition}. More recent work generalised the BV theorem for certain classes of field theories (and gauge fixings) \cite{costello2011renormalization, costello2016factorization}, while a perturbative approach has been shown to work for the theory at hand in the presence of boundaries \cite{cattaneo2018perturbative}, and with a regularisation coming from cellular decompositions in \cite{cattaneo2017cellular}.
\end{remark}

\begin{remark}\label{Lagsub}
Notice that the notion of Lagrangian submanifold has to be appropriately adapted in infinite dimensions and when dealing with $\mathbb{Z}$-grading. The Lagrangian submanifolds $L$ we will consider in this paper are such that, locally, the symplectic space looks like $L\oplus K$, with the symplectic form given by a nondegenerate pairing between $L$ and $K$. This notion of a Lagrangian submanifold coincides with the one used in \cite{cattaneo2018split}.  Often $L$ can be seen as the vanishing locus of a Poisson subalgebra $\mathcal{I}$ of the Poisson algebra of functions on $\mathcal{F}_{BF}$, which is also isotropic, i.e. $\Omega_{BF}\vert_{L} = 0$. This means that $L$ is isotropic and coisotropic\footnote{This notion coincides with requiring $L$ to be maximal isotropic.}. For a more in depth analysis of Lagrangian submanifolds in infinite dimensions and the symplectic category we refer to \cite{weinstein2010symplectic}, building on
\cite{weinstein1971symplectic}.
\end{remark}

\begin{remark}
For concreteness, in what follows we will discuss field theories where fields are given by differential forms on a manifold. If needed, one can think of the space of fields as a Fr\'echet vector space, but indeed this specification will not be necessary for our purposes. 
\end{remark}

\subsection{Partition functions}
If we look at the quadratic part of the action functional, we can interpret the partition function as a (formal) Gaussian integral. Assume from now on that the action functional is at most quadratic. 

In finite dimensions the result of said integral would be the determinant of the operator featured in the action functional $S_M$. In infinite dimensions, this requires defining an appropriate regularisation of determinants.

The standard approach to partition functions for degenerate quadratic functionals follows from Schwarz 
\cite{schwarz1978partition, schwarz1979partition}, 
where the resolution of (the kernel of) an elliptic differential operator on a closed Riemannian manifold is presented, which outputs a (co-)chain complex, and the partition function is given in terms of products of (regularised) determinants of operators associated with the resolving complex. The explicit example for $BF$ theory is given in Section \ref{Sect:ATBF}. In short, the insights from Schwarz allow us to interpret partition functions of quadratic degenerate functionals as (a product of) regularised determinants, provided a suitable resolution can be found such that the regularised determinants of the associated operators exist.

For the purposes of this paper, instead of the standard zeta-function regularisation, we will use the notion of a  {flat-determinant}, based on  {flat-traces}, inspired by Atiyah and Bott's constructions 
\cite{atiyah1964notes, atiyah1968lefschetz}. 
Details on the definition of flat traces and determinant can be found in several places: \cite[Definition 3.12]{batalin1984gauge},
a microlocal version in 
\cite[Section 2.4]{dyatlov2016dynamical},
and a mollifier approach in 
\cite[Section 3.2.2]{baladi2018dynamical}.
Since we will not be concerned with the microlocal analysis of operators, we will avoid discussing the necessary tools to define flat traces, and will only work up to the requirements that they exist for the operators we will consider. In this spirit, we give the following definition.
\begin{definition}\label{flatdet}
Let $A\colon V\to V$ be an operator on an appropriate inner product space, such that the flat trace $\mathrm{tr}^\flat(\mathrm{exp}(-t(A+\lambda)))$ exists for $\lambda \in \mathbb{C}$. We define the flat determinant of $A+\lambda$ to be
\begin{align}
    \log \det{}^\flat(A + \lambda) 
    \coloneqq - \left.\frac{d}{ds}\right|_{s=0} \left[\frac{1}{\Gamma(s)} \intl_0^{\infty} t^{s-1} \tr^\flat(\exp{(-t(A + \lambda))} - \Pi_\lambda)dt \right]
\end{align}
where $\Pi_\lambda$ is the spectral projector on the kernel of $(A + \lambda)$, whenever the integrals converge.
\end{definition}

\begin{remark} Observe that if $A$ is such that $e^{-t(A + \lambda)}$ is trace-class, then $\tr(e^{-t(A + \lambda)})=\tr^\flat(e^{-t(A+\lambda)})$. As a consequence, if the zeta-regularized determinant of $A+\lambda$ exists, it coincides with the flat-determinant:  ${\det}^{\flat}(A+\lambda)=\det(A+\lambda)$. See e.g. \cite[Proposition 6.8]{baladi2018dynamical} for details. 
\end{remark}

Let $A=\mathrm{diag}(B,C)$ be a graded, degree-preserving (block-diagonal) linear map on a finite dimensional graded vector space. A graded Gaussian integral for $\exp{(-\langle y,Ax\rangle})$ returns $\sdet(A)^{-1}=\frac{\mathrm{det}(C)}{\det(B)}$. More generally, if $A_k$ is the $k$-th component of $A$ acting on a graded space with a finite number of nonzero components, each $A_k$ acting on vectors of degree $k$, we get\footnote{Observe that we are considering parities modulo $2$. In principle a graded determinant would return $\prod_{k=0}^n \det(A)^{k(-1)^k}$. We will not make such a distinction in what follows.}
$$
\sdet(A)=\prod_{k=0}^n \det(A)^{(-1)^k}.
$$
For an introduction to Berezinians and odd integration see, for example, 
\cite[Section 3.8]{mnev2019quantum} 
and 
\cite{voronov1991geometric}, while the original notion was introduced in 
\cite{berezin1983introduction,berezin1975supermanifolds}. 

In infinite dimensions, to an operator $A$ on a graded space we can associate a regularised superdeterminant in the same way, but replacing the determinants on the block operators with their flat-regularised versions. We will use this notion to define the partition function of a degenerate quadratic functional, as follows:  
\begin{definition}\label{Def:partitionfunction}
Let $S\colon V\to \mathbb{R}$ of the form $S= \frac12\int_M (x,Ax)$ for some operator $A\colon V\to V$, with $V$ a (possibly graded) vector space endowed with an inner product $(\cdot,\cdot)$. We define the partition function $Z$ of $S$ to be the square root of the flat (super)-determinant of the (graded) operator $A$:
{\begin{equation}\label{Partitionfunctiondef}
    Z(S)= \intl_{V} e^{i S_M} := |{\sdet}^\flat(A)|^{-\frac12}.
\end{equation}
When the field theory has symmetries, we assume that $V$ is further endowed with a $(-1)$ symplectic form. A choice of a Lagrangian submanifold $\mathbb{L}\subset V$ will be called gauge-fixing, and we define the partition function of $S$ in the gauge-fixing $\mathbb{L}$ to be}
\begin{equation}
    Z(S,\mathbb{L}):=Z(S\vert_{\mathbb{L}}).
\end{equation}
\end{definition}
{
\begin{remark}\label{Rem:omitphase}
To obtain a true generalisation of finite dimensional Gaussian integrals one should append to formula \eqref{Partitionfunctiondef} the phase factor $\exp(-i\frac{\pi}{4}\mathrm{sg}(A))$ where $\mathrm{sg}(A)$ is the signature of the operator $A$, appropriately regularised. Since in what follows we will not discuss the phase of partition functions, we will omit this term from the definition.
\end{remark}
}

\begin{remark}
Observe that often one encounters the situation in which $S_M=\int_M (y,Bx)$ for some operator $B$ and $x,y\in V'$ for some space $V'$. Define\footnote{Note that the $V'$ components in $V$ are not considered to have different degrees, i.e. $A$ is an even matrix of graded operators.} $V=V'\oplus V'$ and 
$$
A= \left(
\begin{array}{cc}
0 & B \\
B^t & 0
\end{array}\right)
$$
so that $(y,Bx) = \frac12 (z,Az)$ with $z= (x,y) \in V$. Then $Z=|\sdet^\flat(A)|^{-\frac12} = |\sdet^\flat(B)|^{-1}$.
\end{remark}

\begin{remark}\label{rem:sdetandshift}
Let $V$ be a graded vector space and $V[k]$ its $k$-shift, so that $(V[k])^i := V^{i+k}$. In particular, if $z\in V[1]$ is a homogeneous element of degree $k$ it will be parametrised by homogeneous elements in $V$ of opposite parity. In particular, if $A$ is a graded linear map on $V$, the Gaussian integral 
\begin{equation}
    \int_{V[1]} e^{-\frac12(z,A,z)} := |{\sdet}^\flat (A)|^{\frac12}.
\end{equation}
\end{remark}

\section{Geometric Setting}\label{Sect:geoprel}

This section establishes the geometry and notations. In the following three sections we progressively introduce more structure to our initial setup of a flat vector bundle over a manifold, whose twisted cohomology is trivial. The plainest setting involves a differentiable manifold endowed with a flat vector bundle. On top of that we consider the introduction of either a Riemannian structure or a contact structure. The intersection of the two will require the base manifold to display Anosov dynamics.

It is useful to distinguish between these geometric settings as when we will only need certain geometric properties when discussing different field theories. This distinction between geometric data with which a differentiable manifold is endowed reflects the practice of complementing topological theories with additional geometric structures (e.g. Riemannian or contact) for the sake of gauge fixing.

\subsection{Flat vector bundle}\label{subsec:geo-vanilla}
Let $M$ be an $N$-dimensional compact manifold without boundary which is oriented and connected. Let $\rho:\pi_1(M)\to U(\C^r)$ denote a unitary representation. This representation endows $M$ with a Hermitian vector bundle $(E,h)$ of rank $r$ with flat connection $\nabla$. We collectively denote this data
by $(M,E)$.

Let $\Omega^\bullet(M;E)$ denote the space of (smooth) differential forms on $M$ taking values in $E$, and let
\begin{align}
    d_\nabla\equiv d_k \colon \Omega^k(M,E) \to \Omega^{k+1}(M,E)
\end{align}
be two notations for the twisted de Rham differential. We denote by $H^\bullet(M;E)$
the cohomology associated to the twisted de Rham complex $(\Omega^\bullet(M;E), d_\nabla)$, with Betti numbers $\beta_k:= \dim H^k(M;E)$.

From now on we will assume that $(M,E)$ is such that its twisted de Rham complex is acyclic, i.e. $\beta_k=0$ for all $0\le k \le N$.

\subsection{Riemannian structure and analytic torsion}\label{subsec:geo-riemann}

Suppose $(M,E)$ is supplemented with a Riemannian metric $g_M$. 
We collectively denote this data $(M,E,g_M)$. 
Let $\dVol$ denote the associated volume form and let 
$\langle \cdot, \cdot \rangle$
be the inherited inner product on 
$\bigwedge ^\bullet T^*M$.
The Hodge star 
$\star : \bigwedge^k T^*M \to \bigwedge^{N-k}T^*M$
is defined through 
\begin{align}
    u \wedge \star v = \langle u, v \rangle \dVol
\end{align}
for $u,v\in \bigwedge^k T^*M$. This lifts to $E$-valued forms by identifying $E$ with its dual via the Hermitian metric $h$. An inner product is then placed on $\Omega^\bullet(M;E)$ by declaring
\begin{equation}\label{eq:differential-form-inner-prod}
    (u,v)=\int_{M}[u\wedge\star v]^{\mathrm{top}}
\end{equation} 
for $u,v\in\Omega^\bullet(M;E)$ with $\mathrm{top}$ refers to only taking the top-form (degree $N$) part of $u\wedge \star v$.

The de Rham differential has an adjoint
\begin{align}
    d_\nabla^*\equiv d_k^* \colon \Omega^{k+1}(M,E) \to \Omega^k(M,E)
\end{align}
which provides the twisted Laplace--de Rham  operator (henceforth simply  {Laplacian})
\begin{align}
    \Delta_k := (d_\nabla^* + d_\nabla)^2 : \Omega^k(M;E) \to \Omega^k(M;E).
\end{align}
When acting on 
$L^2(M;\bigwedge^k T^*M \otimes E)$ the Laplacian has nonnegative eigenvalues $\lambda_n\geq 0$. Denote by $\Pi_{\lambda}$ the $L^2$ spectral projector onto the kernel of $\Delta_k+\lambda$ and consider the function
\begin{align}\label{Lapzeta}
	F_{\Delta_k} (\lambda,s)
    &:=
    \frac{1}{\Gamma(s)}	\int_0^\infty t^{s-1} {\tr}^\flat \left( e^{-t(\Delta_k+\lambda)} - \Pi_{\lambda} \right) dt.
\end{align}
The (flat) regularised determinant of the operator $\Delta_k + \lambda$ is
\begin{align}
	\log {\det}^\flat(\Delta_k + \lambda) = - \left.\frac{\partial}{\partial s}\right|_{s=0} F_{\Delta_k}(\lambda, s).
\end{align}
We also write 
$f_{\Delta_k}(s) := F_{\Delta_k}(0,s)$ 
and observe that 
$f_{\Delta_k}(s) = \sum_{\lambda_j>0} {\lambda_j}^{-s}$.

\begin{definition}\label{def:analytic-torsion}
The analytic torsion of $M$ 
\cite{ray1971r}
is defined to be
\begin{align}
	\tau_{\rho}(M)
    :=
    \prod_{k=1}^N  {\det}^\flat(\Delta_k)^{\frac{k}{2}(-1)^{k+1}}.
\end{align}
\end{definition}
Alternatively, we can write
\begin{align}
	2\log \tau_{\rho}(M)
    =
    \sum_{k=1}^{N} (-1)^{k} k \left.\frac{d}{ds}\right|_{s=0} f_{\Delta_k}(s).
\end{align}
The Hodge decomposition provides an orthogonal decomposition of
$L^2(M;\bigwedge^\bullet T^*M\otimes E)$
into exact and coexact forms (no harmonic forms are present due to acyclicity). Introducing 
$d_k^*d_k:=\Delta_k|_{\textrm{coexact}}$ 
we can also write the analytic torsion as
\cite{schwarz1979partition} (see also \cite{mnev2014lecture,cattaneo2017cellular}) 

\begin{align}
	\tau_{\rho}(M)
    =
    \prod_{k=0}^{N-1} {\det}^\flat(d_k^*d_k)^{\frac{1}{2}(-1)^{k}},
\end{align}
which alternatively reads
\begin{align}
	2\log \tau_{\rho}(M)
	= 
	\sum_{k=0}^{N-1} (-1)^{k+1} \left.\frac{d}{ds}\right|_{s=0} f_{d_k^*d_k} (s).
\end{align}

\subsection{Contact structure}\label{subsec:geo-contact}

Suppose that $(M,E)$ is supplemented with a contact form $\alpha\in \Omega^1(M)$, and $\mathrm{dim}(M)= N =2n+1$. We will denote $\dVol = \alpha\wedge (d \alpha)^n$. Let $X$ be the associated Reeb vector field, defined by the relations
\begin{align}
    \iota_X\alpha =1, \qquad \iota_X d\alpha=0.
\end{align}
We collectively denote this data $(M,E,X)$.

Denote by $T^*_0M$ the $2n$-rank subbundle of $T^*M$ defined as the conormal of $X$, so that pointwise $T^*M = \R \alpha + T^*_0M$.
Transferring this to the space of $E$-valued differential forms, we write
\begin{equation}\label{Omegasplitting}
    \Omega^k(M,E)
    =
    \Omega^k_0(M,E) \oplus \alpha\wedge \Omega^{k-1}_0(M,E).
\end{equation}
where $\Omega^k_0(M,E)=\ker \iota_X |_{\Omega^k(M)}$.

\subsubsection{Contact-Riemannian structure} 
Let $(M,E,X)$ as above and introduce a metric $g_M$ on $M$ of the form 
$g_M=\alpha^2+g_0$ 
such that 
$T^*M = \R \alpha + T^*_0M$ 
becomes an orthogonal decomposition. 
Let $\star$ and $\star_0$ denote the Hodge stars associated with $(T^*M,g_M)$ and $(T^*_0M,g_0)$ respectively.
We may choose $g_0$ such that $\star\alpha=(d\alpha)^n=\star_0 1$ whence the Hodge star behaves nicely with respect to the splitting of $T^*M$. Specifically, we have maps
\begin{subequations}\label{eq:starrules}
\begin{align}
	\star : \Omega^{k}_0(M,E) &\to \alpha\wedge\Omega^{N-k-1}_0(M,E)
    &
	\star : \alpha\wedge\Omega^{k}_0(M,E) &\to \Omega^{N-k-1}_0(M,E)
	\\
	\varphi &\mapsto (-1)^{k}\alpha \wedge \star_0\varphi
	&
	\alpha\wedge \psi &\mapsto \star_0 \psi
\end{align}
\end{subequations}
Moreover, noting that 
$\langle \alpha\wedge \varphi, \alpha\wedge\varphi\rangle 
= 
\langle \varphi, \varphi \rangle$ 
on $\Omega^{\bullet}_0(M,E)$ 
we conclude $(\iota_X)^T = \alpha \wedge$.
When a metric is chosen in this way, compatible with the contact stucture, we collectively denote this data $(M,E,X,g_M)$.

\begin{definition}\label{Def:detONE}
Considering the maps 
\begin{align}
    \alpha\wedge \colon \Omega^{\bullet}_0(M,E) \to \alpha\wedge \Omega^{\bullet}_0(M,E),
    &&
    \iota_X: \alpha\wedge \Omega^{\bullet}_0(M,E) \to \Omega^{\bullet}_0(M,E),
\end{align}
we set
\begin{align}
    {\sdet}^\flat(\alpha\wedge) = {\sdet}^\flat(\iota_X):= \left|{\sdet}^{\flat}(\iota_X\circ \alpha\wedge)^{\frac12}\right| = 1.
\end{align}
\end{definition}

\subsection{Anosov dynamics}\label{subsec:geo-anosov}

Suppose $M$ is supplemented with a flow $\varphi_t:M\to M$ for $t\in\R$. We will reuse $X\in C^\infty(M;TM)$ to denote the vector field which generates $\varphi_t$.

\begin{definition}\label{def:Anosov}
The flow is  {Anosov} if there exists a $d\varphi_t$-invariant continuous splitting of the tangent bundle:
\begin{align}
    T_xM = E_n(x) \oplus E_s(x) \oplus E_u(x),
    &&
    E_n(x)=\R X_x,
\end{align}
and for a given norm $\|\cdot \|$ on $TM$, there exist constants $C, \lambda>0$ so that for all $t\ge 0$,
\begin{align}
    \forall v \in E_s(x), \quad \| d\varphi_t(x)v\| \le C e^{-\lambda t}\|v\|,
    \qquad
    \forall v \in E_u(x), \quad \| d\varphi_{-t}(x)v\| \le C e^{-\lambda t}\|v\|.
\end{align}
The subbundles $E_n, E_s, E_u$ are respectively called neutral, stable, unstable. 
\end{definition}

\begin{remark}\label{remark:cotangent-bundle-anosov-decomposition}
We will prefer to work with the cotangent bundle $T^*M$ due to the Lie derivative acting naturally on differential forms. The cotangent bundle also has a decomposition $T^*_xM=E_n^*(x) \oplus E_s^*(x) \oplus E_u^*(x)$ whose stable and unstable bundles are understood through the action of $\left(d\varphi_{-t}\right)^T$ (rather than $d\varphi_t$).
\end{remark}

Henceforth we will always assume that the stable and unstable bundles are orientable and each have rank $n$.

\subsubsection{A guiding example}\label{subsec:geo-example}

We provide an example of the geometric setting discussed in the previous subsections.
Let $(\Sigma, g)$ be a compact manifold without boundary which is oriented, connected, and of dimension $n+1$.
Suppose that $\Sigma$ has sectional curvature which is everywhere strictly negative.
Let $M:=S^*_g\Sigma$ be the unit cotangent bundle of $\Sigma$.
Set $\alpha \in \Omega^1(M)$ to be the pull-back of the canonical one-form on $T^*\Sigma$. Then $(M,\alpha)$ is a contact manifold, 
and the Reeb vector field $X_g\in C^\infty(M;TM)$ generates 
the geodesic flow $\varphi_t$ which is Anosov
\cite{anosov1967geodesic, anosov1967some, arnold1968problemes}. 

If we consider $M=S_g^*\Sigma$ within the geometric setting $(M,E)$ of Subsection \ref{subsec:geo-vanilla} (in particular the representation $\rho$ is unitary and $\nabla$ is flat), we denote the resulting contact, Anosov, Riemannian data on $S_g^*\Sigma$ by 
{
$$(M=S_g^*\Sigma, E, X_g, g).$$
}

{
\begin{remark}
$M\to \Sigma$ is an $\S^{n}$-bundle, and if $n\ge 2$ then $\pi_1(\S^{n})=0$ whence representations of $\pi_1(\Sigma)$ are in one-to-one correspondence with representations of $\pi_1(M)$. 
\end{remark}
}

{
\begin{remark}\label{rem:squaretorsion}
Observe that in the case of a unitary representation and a flat vector bundle, and when $n\ge 2$, the preceding remark implies that $\tau_\rho(M)=\left(\tau_\rho(\Sigma)\right)^2$
\cite[Section 1, p. 526]{fried1986analytic}. For the rest of this article, particularly the announcements of theorems and conjectures, we have chosen to emphasise the role of $M=S^*_g\Sigma$ rather than $\Sigma$.
\end{remark}
}

\section{Ruelle Zeta function} \label{s:RZF}

Consider the geometric data 
$(M,E,X)$
of Section~\ref{subsec:geo-contact}
and assume the flow $\varphi_t$ associated to $X$ is Anosov (see Section \ref{subsec:geo-anosov}).
We denote by $\mathcal{P}$ the set of primitive orbits of the flow $\varphi_t$ and by $\ell(\gamma)$ the period of any given $\gamma\in \mathcal{P}$. 
\begin{definition}
The Ruelle zeta function (associated with the trivial representation of $\pi_1(M)$) is defined as
\begin{align}
    \zeta(\lambda) := \prod_{\gamma \in \mathcal{P}} (1 - e^{-\lambda \ell(\gamma)})
\end{align}
whose convergence is assured for $\Re\lambda \gg 1$. The Ruelle zeta function twisted by an arbitrary representation $\rho$ is
\begin{align}\label{ruellezetafunction}
    \zeta_\rho (\lambda) := \prod_{\gamma \in \mathcal{P}} \det(I - \rho([\gamma])e^{-\lambda \ell(\gamma)}).
\end{align}
whose convergence is again assured for $\Re\lambda \gg 1$.
\end{definition} 
Here $[\gamma]$ represents the conjugacy class of $\gamma$ in $\pi_1(M)$. It has been shown that the zeta functions continue meromorphically to $\C$ \cite{butterley2007smooth, marklof2004selberg, giulietti2013anosov, dyatlov2016dynamical}. We have the following:

\begin{theorem}[\cite{fried1986analytic}]\label{thm:FCT}
Let {$(M,E,X,g)$ be the geometric data of Subsection \ref{subsec:geo-example}, where $M=S^*_g\Sigma$, with $\Sigma$ a closed, oriented hyperbolic manifold $\Sigma=\Gamma\backslash \mathbb{H}^{n+1}$}, and $g$ the induced hyperbolic metric. Then, the Ruelle zeta function, defined for $\mathrm{Re}(s)>n$ by equation \eqref{ruellezetafunction} extends meromorphically to $\mathbb{C}$ and

$$
{|\zeta_\rho(0)|^{(-1)^{n}} =\tau_\rho(M).}
$$

\end{theorem}

This fact has inspired a conjecture \cite{fried1987lefschetz}:

{\begin{conjecture}[Fried]\label{conj:fried}
Let $(M,E,X,g)$ be the geometric data of Subsection \ref{subsec:geo-example} with $\mathrm{dim}(M)=2n+1$. Then the (twisted) Ruelle zeta function computes the analytic torsion:
	\begin{align}\label{eqn:fried}
		|\zeta_\rho (0)| ^{(-1)^{n}}  = \tau_\rho (M).
	\end{align}
\end{conjecture}}

\subsection{Differential forms decomposition}

This section shows how the (twisted) Ruelle zeta function may be written as an alternating product of zeta functions associated with
$\left(\bigwedge^k T^*_0M\right)\otimes E$ for $0\le k \le 2n$.

Given a closed orbit $\gamma$, of length $\ell(\gamma)$, and a point $x\in M$ in the orbit, let $P(\gamma,x)$ denote the linearised Poincar\'e map on the fibre of $T^*_0M$ above $x$:
\begin{align}
	P(\gamma,p) := \left( d\varphi_{-\ell(\gamma)} \right)^T : T^*_0M_{(x)} \to T^*_0M_{(x)}
\end{align}
This map is conjugate to $P(\gamma, x')$ for other $x'$ in the same orbit $\gamma$, and as we need only evaluate this map's trace and determinant, we will refer to all maps as $P(\gamma)$.
Recall $T^*_0M$ is the subbundle of $T^*M$ conormal to the vector field $X$.

We start with the basic linear algebra identity 
$$\det(I-A) = \sum_{k=0}^{m} (-1)^k \tr (\wedge^k A)$$ 
for an endomorphism $A$ on an $m$-dimensional vector space and the observation that $P(\gamma)$ has precisely $n$ eigenvalues greater than 1 (where $n$ is the rank of $E_u^*$). Therefore for $j\ge 1$
\begin{align}\label{linalgid}
    (-1)^n | \det(I - P(\gamma)^j) | 
    =
    \sum_{k=0}^{2n}
    (-1)^k \tr (\wedge^k P(\gamma)^j).
\end{align}
Using identity \eqref{linalgid} to obtain equality \eqref{e:insertlinalgid} below, we derive
\begin{subequations}
\begin{align}
    \log \zeta_\rho(\lambda)
    &= \sum_{\gamma \in \mathcal{P}} \tr \log (I - \rho([\gamma])e^{-\lambda \ell(\gamma)}) \\
    &= -\sum_{\gamma \in \mathcal{P}} \tr \sum_{j=1}^\infty \frac1j e^{-\lambda j \ell(\gamma)} \rho([\gamma])^j \\ \label{e:insertlinalgid}
    &= -\sum_{\gamma \in \mathcal{P}} \sum_{j=1}^\infty \frac1j e^{-\lambda j \ell(\gamma)} \tr(\rho([\gamma])^j) 
    \left( (-1)^{n}\sum_{k=0}^{2n}
    \frac{(-1)^k \tr (\wedge^k P(\gamma)^j)}
    {| \det(I - P(\gamma)^j) |}
    \right) \\\label{e:formdecomp}
    &= (-1)^n\sum_{k=0}^{2n} (-1)^k \log \zeta_{\rho,k}(\lambda).
\end{align}\end{subequations}
Equation \eqref{e:formdecomp} in the preceding display defines implicitly the Ruelle zeta function twisted by $\rho$ upon restriction to $k$-forms 
\begin{align}
    \zeta_{\rho,k}(\lambda) := \exp \left(
    -\sum_{\gamma\in\mathcal{P}} \sum_{j=1}^\infty \frac1j
    \frac{ e^{-\lambda j \ell(\gamma)} \tr(\rho([\gamma])^j) \tr (\wedge^k P(\gamma)^j)}
    {| \det(I - P(\gamma)^j) |}
    \right)
\end{align}
so in the form of a compact equality, we have
\begin{equation}\label{Zetadecomposition}
\zeta_{\rho}(\lambda)^{(-1)^n } = \prod_{k=0}^{2n} \zeta_{\rho,k}(\lambda)^{(-1)^k}.
\end{equation}

\subsection{Ruelle zeta function as regularized determinant}\label{subsec:Ruelle zeta function as regularized determinant}
We aim to link the function $\zeta_{\rho,k}$ with the operator $\Lie_{X,k}$ acting on $\Omega^k_0(M;E)$.
This is done with the help of the Atiyah--Bott--Guillemin trace formula 
\cite{guillemin1977lectures}. 
We use the notation
\begin{align}
	e^{-t\Lie_{X,k}}={\varphi_{-t}}^* : 
	    L^2(M; \bigwedge^kT^*_0M\otimes E)
	    \to 
	    L^2(M; \bigwedge^kT^*_0M\otimes E).
\end{align}
and write the Schwartz kernel $K_k(t,\cdot,\cdot)$ such that 
\begin{align}
	(e^{-t\Lie_{X,k}} \psi)(x) = \int_M K_k(t,x,y) \psi(y) dy.
\end{align}
Due to the microlocal structure of $K_k$ we may take its flat trace, which leads to the Atiyah-Bott-Guillemin trace formula:
\begin{equation}\label{flattrace}
    \tr^\flat e^{-t\Lie_{X,k}}
    = \sum_{\gamma\in \mathcal{P}} \sum_{j=1}^\infty
    \ell(\gamma) \delta(t - j\ell(\gamma))
    \frac{\tr(\rho([\gamma])^j) \tr (\wedge^k P(\gamma)^j)}
    {| \det(I - P(\gamma)^j) |}
\end{equation}
as a distribution on $\R_+$. Integrating this distribution against the function $-t^{-1}e^{-t\lambda}$ provides an integral representation of $\zeta_{\rho,k}(\lambda)$:
\begin{equation}\label{integralrepzeta}
    \log \zeta_{\rho, k}(\lambda) 
    = -\int_0^\infty 
    t^{-1}e^{-t\lambda} \tr^\flat e^{-t\Lie_{X,k}}
    \,dt.
\end{equation}

Consider now the following function, dependent on two variables:
\begin{align}
F_{\Lie_{X,k}}(\lambda,s)
&:= \frac{1}{\Gamma(s)} \int_0^\infty t^{s-1} \tr^\flat e^{-t(\Lie_{X,k}+\lambda)} dt.
\end{align}
This function is holomorphic for small $s$ (near $s=0$) and for $\Re(\lambda)\gg 1$
(see Subsection~\ref{subsec:mero-extn-resolvent}). 
Naively, at $s=0$ the integrand poses a problem due to the $t^{-1}$ structure as $t\to0$, however the trace formula (Equation~\eqref{flattrace}) provides a natural cut-off in small $t$ so that we avoid this problem. Moreover, for small $s$ we expand $1/\Gamma(s)=s+O(s^2)$ showing
\begin{align}
\left. \partial_s \right|_{s=0} F_{\Lie_{X,k}}(\lambda,s)
&=
\left(\left. \partial_s \right|_{s=0} \frac{1}{\Gamma(s)} \right)
\cdot
\left. \int_0^\infty t^{s-1} \tr^\flat e^{-t(\Lie_{X,k}+\lambda)} dt \right|_{s=0} \\
&= - \log \zeta_{\rho, k}(\lambda).
\end{align}
Recalling Definition~\ref{flatdet} for the flat determinant of an operator now indicates that, for $\Re\lambda \gg 1$ 
\begin{align}
    \log {\det} ^\flat (\Lie_{X,k} + \lambda) 
    \equiv 
    - \left. \partial_s \right|_{s=0} F_{\Lie_{X,k}}(\lambda,s)
    =
    \log \zeta_{\rho, k}(\lambda).
\end{align}
Since $\zeta_{\rho,k}(\lambda)$ has a meromorphic extension to the complex plane, so does the function $\det^{\flat}(\mathcal{L}_{X,k} + \lambda)$. Assuming there are no poles at $\lambda=0$, we sensibly have $\det^{\flat}(\mathcal{L}_{X,k})$ as the value at zero of the meromorphic extension of the Ruelle zeta function (restricted to $k$-forms) at zero: 
\begin{equation}\label{eqn:det-L=zeta-0}
    {\det}^{\flat}(\mathcal{L}_{X,k}) = \zeta_{\rho,k}(0).
\end{equation}
The decomposition in
Equation \eqref{Zetadecomposition}
of the Ruelle zeta function now gives an operator interpretation of the Ruelle zeta function: 
\begin{equation}\label{Ruelleasadeterminant}
    \zeta_\rho (\lambda)^{(-1)^n} = \prod_{k=0}^{2n} {\det}^\flat (\Lie_{X,k} + \lambda) ^{(-1)^k}.
\end{equation}

\begin{remark}\label{rem:zetafunX}
Notice that, due to the decomposition \eqref{Zetadecomposition} and the trace formula \eqref{flattrace}, we can consider $\zeta_\rho$ as directly dependent on an Anosov vector field $X$. To stress this fact we will use the notation $\zeta(X,\lambda)$ (resp. $\zeta_k(X,\lambda)$) instead of $\zeta_\rho(\lambda)$ (resp. $\zeta_{\rho,k}(\lambda))$. 
\end{remark}

\subsection{Meromorphic extension of the resolvent}
\label{subsec:mero-extn-resolvent}
In the preceding section we related the resolvent $(\Lie_{X,k} + \lambda)^{-1}$ with $\zeta_{\rho,k}(\lambda)$ for $\lambda\gg 1$. Here we announce a more precise statement for the meromorphic extension of $(\Lie_{X,k} + \lambda)^{-1}$
\cite{dyatlov2016dynamical}.

The operator norm of $e^{-t\Lie_{X,k}}$ is bounded by $e^{C_0t}$ for some $C_0>0$. Therefore the resolvent $(\Lie_{X,k}+\lambda)^{-1}$, as an operator on $L^2$ sections, exists for $\Re\lambda>C_0$ and is given by the formula
\begin{align}\label{eqn:lie-derivative-inverse-large-lambda}
	(\Lie_{X,k} + \lambda)^{-1} = \int_0^\infty e^{-t(\Lie_{X,k} + \lambda)} dt.
\end{align}
The restricted resolvent
\begin{align}\label{eqn:restricted-resolvent}
	R_k(\lambda)
    =
    (\Lie_{X,k}+\lambda)^{-1}
    :
    C^\infty(M;\bigwedge^k T^*_0M \otimes E)
    \to
    \mathcal{D}'(M;\bigwedge^k T^*_0M \otimes E)
\end{align}
has a nowhere-vanishing meromorphic continuation to $\C$ whose poles are of finite rank, and are called  {Pollicott-Ruelle resonances}.
For each $\lambda_0\in \C$, we have the expansion
\begin{align}
	R_k(\lambda)
	=
	R_k^H(\lambda) + \sum_{j=1}^{J(\lambda_0)} \frac{(-1)^{j-1}(\Lie_{X,k}+\lambda_0)^{j-1}\Pi_{\lambda_0} }{(\lambda-\lambda_0)^j}
\end{align}
where $R_k^H$ is holomorphic near $\lambda_0$ and
$
	\Pi_{\lambda_0}
    :
    C^\infty(M; \bigwedge^k T^*_0M \otimes E) 
    \to 
    \mathcal{D}'(M; \bigwedge^k T^*_0M \otimes E)
$
is a finite rank projector. 
The range of $\Pi_{\lambda_0}$ defines  {(generalised) resonant states}. They are characterised as
\begin{align}\label{eqn:resonant-states-characterisation}
    \operatorname{Range} \Pi_{\lambda_0}
    =
    \{ \varphi\in\mathcal{D}'(M,\bigwedge^k T^*_0M \otimes E) 
    : 
    \operatorname{WF}(u)\subset E^*_u, 
    (\Lie_{X,k}+\lambda_0)^{J(\lambda_0)}\varphi=0 
    \}.
\end{align}
where $\operatorname{WF}$ refers to the wave-front of a distribution (or current), $E_u^*$ is the unstable bundle referenced in Remark~\ref{remark:cotangent-bundle-anosov-decomposition}, and $J(\lambda_0)$ denotes the multiplicity. The adjective ``generalised" refers to the possibility that the pole may not be simple (and is superfluous in the case $J(\lambda_0)=1$). 

Finally, the poles of the meromorphic continuation $R_k(\lambda)$ correspond to zeros of the zeta function $\zeta_k(\lambda)$ (which is entire for each $k$), and the rank of the projector $\Pi_{\lambda}$ equals the multiplicity of the zero.

\section{$BF$ theory on contact manifolds}\label{s:BF}
In this section we will analyse a field theory that goes under the name of $BF$ theory\footnote{The name  $BF$ comes from the tradition of denoting fields with $B$ and $A$, and Lagrangian density $B\wedge F_A$.}, in the Batalin-Vilkovisky (BV) formalism. Consider the basic geometric data $(M,E)$ of Subsection~\ref{subsec:geo-vanilla}. The  {classical version} of abelian $BF$ theory (i.e. without BV Formalism) is given by the following assignment:
\begin{definition}
Define the space of classical fields to be $F_{BF}\coloneqq\Omega^1(M,E)\oplus \Omega^{N-2}(M,E)$, and the BF action functional:
\begin{equation}\label{e:classicalBFaction}
    S_{BF} = \intl_{M} B\wedge d_\nabla A.
\end{equation}
Then we call  {$BF$ theory} the assignment $(M,E)\leadsto (F_{BF},S_{BF})$.
\end{definition}

\begin{remark}
Recall that Subsection~\ref{subsec:geo-vanilla} equips $E$ with a Hermitian metric $h$. This metric has been implicitly used to couple $B$ and $A$. Specifically, both $B$ and $d_\nabla A$ take values in $E$ whence $h$ must be used so that $B\wedge d_\nabla A \in \Omega^N(M)$.
\end{remark}

\begin{remark}
It is easy to see that shifting either $B$ or $A$ by a $d_\nabla$-exact form leaves the action functional unchanged.\footnote{Strictly speaking this is true only up to boundary terms. We will assume that $M$ has no boundary, but otherwise the analysis of boundary terms is relevant, and can be performed with a version of the BV formalism \cite{cattaneo2018perturbative}.}
This goes under the name of  {reducible symmetry}, and it is conveniently treated by means of the BV formalism.
\end{remark}

Let us consider the space of differential forms $\Omega^{-\bullet}(M,E)$ as a $\mathbb{Z}$-graded vector space, such that homogeneous forms $\omega^{(k)}\in\Omega^k(M,E)$ will have degree $|\omega^{(k)}|:=-k$. We define the space of Batalin--Vilkovisky fields for $BF$ theory to be the graded vector space
\begin{equation}
    \mathcal{F}_{BF}:=
        \Omega^{-\bullet}(M,E)[1]\oplus\Omega^{-\bullet}(M,E)[N-2]\ni(\A,\B)
\end{equation}  
where the degree shift means that a $k$-form in $\Omega^{-\bullet}(M,E)[1]$ will have degree $|\A|=1-k$. The symplectic structure reads
\begin{equation}
    \Omega_{BF} = \int\limits_M [\delta\B \delta \A]^{\mathrm{top}},
\end{equation}
and we define an  {action} functional on $\mathcal{F}_{BF}$ as
\begin{equation}\label{e:BVBFaction}
    \mathbb{S}_{BF} = \int\limits_M [\B d_\nabla \A]^{\mathrm{top}}.
\end{equation}
Observe that $|\mathbb{S}_{BF}|=0$ and $|\Omega_{BF}|=-1$.

\begin{remark}
We stress that, although the  {functional form} of $S_{BF}$ in Equation \eqref{e:classicalBFaction} and $\mathbb{S}_{BF}$  in \eqref{e:BVBFaction} is the same, in the latter $\B$ and $\A$ are inhomogeneous forms with an additional shift in degree. Such a degree shift effectively switches the total parity of inhomogeneous forms, so that, if $N$ is odd, even forms will have odd parity and vice-versa. This will have a crucial impact in the partition-function interpretation of quantities such as the analytic torsion and Ruelle zeta function.
\end{remark}

\begin{definition}\label{Def:BVBF}
The assignment $(M,E)\leadsto (\mathcal{F}_{BF}, \Omega_{BF}, \mathbb{S}_{BF}, Q_{BF})$, with 
\begin{align}
    Q_{BF}\B=d_\nabla\B,
    \qquad 
    Q_{BF}\A=d_\nabla\A
\end{align}
is called  {$BF$ theory in the Batalin Vilkovisky formalism}.
\end{definition}

\subsection{Analytic torsion from resolutions of de Rham differential}\label{Sect:ATBF}
In this section we will discuss the relation of analytic torsion with degenerate action functionals, following Schwarz \cite{schwarz1978partition, schwarz1979partition}. This is related to BF, as we will highlight in what follows. In Section \ref{s:BVinterpretation} we will interpret this relation in terms of a gauge-fixing for $BF$ theory in the Batalin--Vilkovisky framework. 

Consider the geometric data $(M,E, g_M)$ of Subsection~\ref{subsec:geo-riemann}. Note that the only importance of the vector bundle is to ensure acyclicity of the twisted de Rham complex. This requirement will be necessary in Section \ref{s:homotopies}, however, what we will say here can be extended to nontrivial cohomology along the lines of \cite{cattaneo2017cellular,cattaneo2018perturbative}.

We can define the partition function associated to abelian $BF$ theory as follows. We first need a resolution of the kernel of the operator featuring in the action functional, represented then as the $0$th-cohomology of a chain complex. For abelian $BF$ theory it is given by the following: 
\begin{equation}\label{eqon:resolution}
	0\rightarrow V_{N-2}\xrightarrow{T_{N-2}} V_{N-1}\xrightarrow{T_{N-1}}\dots\xrightarrow{T_3} V_2\xrightarrow{T_2}V_1\xrightarrow{T_1}V_0\xrightarrow{T}V_0\rightarrow 0
\end{equation} where  
\begin{subequations}\label{Schwartzmaps}
\begin{align}
    V_0&=\Omega^1(M)\oplus\Omega^{N-2}(M),
    & 
    T&=\begin{bmatrix}
            0  & \star d_{N-2}\\
            \star d_1 & 0
        \end{bmatrix};
    \\
    V_1&= \Omega^0(M)\oplus\Omega^{N-3}(M), 
    &
    T_1&=\begin{bmatrix}
        d_0  & 0\\
        0  & d_{N-3}
        \end{bmatrix};
    \\
    V_2&=\Omega^{N-4}(M), 
    &
    T_2&=\begin{bmatrix}
        0\\
        d_{N-4}
        \end{bmatrix};
    \\ 
    V_k&=\Omega^{N-(k+2)}(M),
    &
    T_k&=d_{N-(k+2)}
\end{align} 
\end{subequations}
for $3\leq k\leq N-2$. 
This implies 
$T_kT_k^*=d_{N-(k+2)}d_{N-(k+2)}^*$
and 
\begin{align}
    T^2=\begin{bmatrix}
        d_1^*d_1  & 0\\
        0         & d_{N-2}^*d_{N-2}
        \end{bmatrix}, 
    \qquad
    T_1T_1^*=\begin{bmatrix}
            d_0d_0^*  & 0\\
            0         & d_{N-3}d_{N-3}^*
            \end{bmatrix}, 
    \qquad 
    T_2T_2^*=\begin{bmatrix}
                0   & 0\\
                0   & d_{N-4}d_{N-4}^*  
            \end{bmatrix}.
\end{align}

In fact, denoting $C:=(A,B)\in V_0$, the operator $T$ allows us to rewrite 
\begin{equation}\label{eq:degenerate actional functional}
    S_{BF} 
    =
    \intl_{M} B\wedge dA
    =
    \frac{1}{2}\left(C,TC \right)
\end{equation} where $(\,,\,)$ is the inner product \eqref{eq:differential-form-inner-prod} on $V_0$ as discussed in Section \ref{subsec:geo-riemann}. Then, the partition function of the degenerate action functional \eqref{eq:degenerate actional functional}, with respect to the resolution \eqref{eqon:resolution}, is defined by Schwarz in \cite{schwarz1979partition} to be:
\begin{equation}\label{eq:metric_gauge_pf}
Z_{\text{Sch}}[T,T_i]:={\det}^{\flat}(T^2)^{-\frac{1}{4}}\prod_{k=1}^{N-2}{\det}^{\flat}(T_kT_k^*)^{(-1)^{k+1}\frac{1}{2}}.
\end{equation}

\begin{remark} \label{rem:Schwarzprocedure}
The procedure outlined above recursively produces spaces $V_k$ at every stage in order to parametrise the relevant quotients $V_k/\mathrm{ker}(T_k)$, on which integration makes sense. Indeed, said quotients coincide with $\mathrm{coker}(T_{k+1})\simeq \mathrm{ker}(T_{k+1}^*)$, and the localisation to the correct integration subspaces is obtained by restriction to $\mathrm{ker}(T^*_k)$. Thus, Schwarz's definition of the partition function goes through an extension to a larger space of fields, and the choice of a subspace where integration is well defined: the kernel of the map $\mathbb{T}\coloneqq\oplus_{j=1}^{N-2}T_j^*$. Notice that this construction produces a resolution in the sense of Equation \eqref{KTres}.
\end{remark}

In the case of abelian $BF$ theory this leads to
\begin{multline}\label{expression:partitionfunction}
	Z_{\text{Sch}}[T,T_i]=
	{\det}^\flat(d_1^*d_1)^{-\frac{1}{4}}{\det}^\flat(d_{N-2}^*d_{N-2})^{-\frac{1}{4}}{\det}^\flat(d_0^*d_0)^{\frac{1}{2}}\\
	\times{\det}^\flat(d_{N-3}^*d_{N-3})^{\frac{1}{2}}
		\prod_{k=2}^{N-2}{\det}^\flat(d_{N-(k+2)}d_{N-(k+2)}^*)^{\frac{(-1)^{k+1}}{2}}.
\end{multline}

It is possible rewrite the expression for the partition function in a familiar form using the following lemma. 
\begin{lemma}\label{lemma:det_relations} 
Let $(M,E,\nabla,\rho)$ be as in Subsection \ref{subsec:geo-vanilla}, then the following holds:
\begin{enumerate}
\item  ${\det}^\flat(d_k^* d_k)={\det}^\flat(d_k d_k^*)$
\item ${\det}^\flat(d_{k-1} d_{k-1}^*)= {\det}^\flat(d_{N-k}^*d_{N-k})$
\item ${\det}^\flat(\Delta_k)={\det}^\flat(d_{k-1}d_{k-1}^*){\det}^\flat(d_{k}^*d_{k})={\det}^\flat(d_{k-1}^*d_{k-1}){\det}^\flat(d_{k}^*d_{k}).$
\end{enumerate}
\end{lemma}
\begin{proof}[Sketch of proof]
The proof is immediate from the analysis of the spectra of the operators under the consideration. More precisely, note that in this case ${\det}^{\flat}$s are spectral invariants as they are the usual zeta-regularized determinants. 

(1) follows from the fact that the operators $d_k^*d_k$ and $d_kd_k^*$ are isospectral, a property one can check by acting with $d$ on a coexact eigenform of $\Delta$. 

Similarly, we observes that $(\ast_{N-k})^{-1}\circ d_{k-1} d_{k-1}^*\circ \ast_{N-k}=d_{N-k}^*d_{N-k}$, which implies that $ d_{k-1} d_{k-1}^*$ and $d_{N-k}^*d_{N-k}$ are isospectral and (2) follows as well. 

Finally (3) follows from the observation that the spectrum of $\Delta_k=d^*_kd_k + d_{k-1}d^*_{k-1}$ is the union of spectrum of $d_{k-1}d_{k-1}^*$ and $d_k^*d_k$.
\end{proof}
Thanks to the previous lemma, we get:
\begin{proposition}
 The partition function $Z_{\text{Sch}}[T,T_i]$ for the resolution \eqref{eqon:resolution} yields
\begin{align}\label{Eq:PFAT}
    Z_{\text{Sch}}[T,T_i]=\tau_\rho(M).
  \end{align}
  \end{proposition}
 \begin{proof}  
 From  (1) and (2) of Lemma ~\ref{lemma:det_relations}, we have $${\det}^\flat\left(d_{N-(k+2)}d_{N-(k+2)}^*\right)={\det}^\flat(d_{k+1}^*d_{k+1}).$$ 
 Using this relation and Lemma \ref{lemma:det_relations} again, we can rewrite (\ref{expression:partitionfunction}) as 
 \begin{multline}
 Z_{\text{Sch}}[T,T_i]=
 {\det}^\flat(d_1^*d_1)^{-\frac{1}{4}}{\det}^\flat(d_{N-2}^*d_{N-2})^{-\frac{1}{4}}{\det}^\flat(d_0^*d_0)^{\frac{1}{2}}\\
	{\det}^\flat(d_{N-3}^*d_{N-3})^{\frac{1}{2}}
		\prod_{k=2}^{N-2}{\det}^\flat(d_{k+1}^*d_{k+1})^{\frac{(-1)^{k+1}}{2}}.
 \end{multline} 
Now, we know again by Lemma ~\ref{lemma:det_relations} that ${\det}^\flat(d_1^*d_1)={\det}^\flat(d_{N-2}^*d_{N-2})$ and ${\det}^\flat(d_{N-3}^*d_{N-3})={\det}^\flat(d_{2}^*d_{2}).$

Finally, using (3) of the Lemma ~\ref{lemma:det_relations}, we can write 
\[\prod_{k=0}^{N-1}{\det}^\flat(d_k^*d_k)^{(-1)^{k}\frac{1}{2}}=
\prod_{k=1}^N{\det}^\flat(\Delta_k)^{(-1)^{k+1}\frac{k}{2}}
\]
In summary, we have 
\[
Z_{\text{Sch}}[T,T_i]=\prod_{k=1}^N{\det}^\flat(\Delta_k)^{(-1)^{k+1}\frac{k}{2}},
\]
showing that the partition function $Z_{\text{Sch}}[T,T_i]$ for the resolution \eqref{eqon:resolution} of the operator $d$ coincides with the analytic torsion $\tau_\rho(M)$ of Definition~\ref{def:analytic-torsion}.
\end{proof}

\subsection{BV interpretation: metric gauge for $BF$ theory}\label{s:BVinterpretation}
We would like to phrase the procedure outlined in Section \ref{Sect:ATBF} in the Batalin--Vilkovisky framework.

First of all we need to observe that, when the complex is acyclic, the submanifold 
\begin{align}
    \mathbb{L}_g\colon\{d_\nabla^{\star} \B = d_\nabla^{\star}\A=0\}
\end{align}
is  Lagrangian in the space of BV fields $\mathcal{F}_{BF}$ (see Remark \ref{Lagsub}), and can be therefore used as a gauge fixing for $BF$ theory. This is a consequence of Hodge decomposition (on closed manifolds without boundary). In fact one can easily check that the images of $d$ and $d^*$ are both isotropic in $\mathcal{F}_{BF}$, and complementary due to Hodge decomposition. We refer to this gauge fixing as the  {metric gauge}.

Due to the vanishing of the cohomology of $\Omega^\bullet(M,E)$, the metric gauge Lagrangian submanifold $\mathbb{L}_g$ can be parametrised as follows.
The condition $d_\nabla^*\B=0$ can be equivalently written as $\B=d_{\nabla}^* \eta$, and if we specify the form degree of the various fields with a subscript, we write $\B_{N-k-1}=d_\nabla^* \eta_{N-k}$.
Finally, writing $\eta_{N-k}=\star \tau_{k}$, together with the inner product given in Equation~\ref{eq:differential-form-inner-prod} provides the following formula:
\begin{equation}\label{Eq:metricgaugeBF}
    \mathbb{S}_{BF}\vert_{\LL_g} 
        \equiv \intl_{M} \left[\B d_\nabla\A \right]^{\mathrm{top}}_{\mathbb{L}_g}
        =\sum_{k=0}^{N-1}\intl_{M}\star d_\nabla\tau_k d_\nabla \A_k\vert_{d^*\A=0}
        =  \sum_{k=1}^N \left( \tau_k, d_\nabla^* d_\nabla\vert_{\text{coexact}} \A_k\right).
\end{equation}

\begin{remark}\label{Rem:ATfromBerezinian}
We consider an interpretation of Equation \ref{Eq:PFAT} in terms of the BV construction we just outlined\footnote{The authors would like to thank P. Mnev and A.S. Cattaneo for valuable insight on this interpretation.}. We will see how the analytic torsion can be seen as a way of making sense of the (super) determinant of the (graded) operator $d\colon \Omega^\bullet_{\text{coexact}} \to \Omega^{\bullet+1}_{\text{exact}}$. Let us define first
\begin{equation}\label{e:detddef}
    {\sdet}^\flat(d_\nabla^*) \equiv {\sdet}^\flat(d_\nabla) \coloneqq {\sdet}^\flat(d_\nabla^*d_\nabla)^{\frac12}.
\end{equation}

Definition \ref{Def:partitionfunction} interprets partition functions of quadratic functionals in terms of the super determinant of the operator that features in the functional. However, when an explicit gauge fixing is considered, one needs to take into account the parametrisaton of the gauge-fixing Lagrangian. In our case this introduces a ``Jacobian superdeterminant'' for the change of coordinates $\B=d_\nabla^*\eta$. Moreover, one should pay attention to the shift in degree $\Omega^{-\bullet}(M,E)[1]$ and $\Omega^{-\bullet}(M,E)[N-2]$. As a matter of fact, in this circumstance the change of variables operator $d^*$ is acting on $k$-form of parity $k+1 \mod 2$ (instead of $k\mod 2$), so that the Jacobian super determinant will appear with the ``opposite" power. For the same reason, recalling Remark \ref{rem:sdetandshift}, the output of a Gaussian integration over a $1$-shifted graded vector space will yield the superdeterminant of the relevant operator (instead of its inverse). 

Hence, one recovers the analytic torsion formally as a Gaussian integral on an super vector space by defining the partition function for gauge-fixed $BF$ theory to be a flat-regularised Berezinian/super determinant, times the Jacobian (super) determinant $\sdet^\flat(d^*)^{-1}$ of the change of variables operator:
\begin{subequations}\label{heuristicberezininangauge}
\begin{align}
    Z(\mathcal{S}_{BF},\,\mathbb{L}_g)
    &= 
    {\int\limits_{\mathbb{L}_g} e^{i\mathbb{S}_{BF}}\vert_{\mathbb{L}_g} 
    = 
    |{\sdet}^\flat(d_\nabla^*)|^{-1}\int \exp{\left(i\sum_{k=0}^{N-1} \left( \tau_{k}, (d_{k}^* d_{k})\A_{k}\right)\right)}}\\
    &:= 
    |{\sdet}^\flat(d_\nabla)|^{-1} \left|{\mathrm{sdet}}^{\flat}(d_\nabla^* d_\nabla)\right| 
    =
    \left|{\mathrm{sdet}}^{\flat}(d_\nabla^*d_\nabla)\right|^{\frac12}\\ 
    &= 
    \prod_{k=0}^{N-1}{\det}^\flat(d_k^*d_k)^{(-1)^{k}\frac{1}{2}} 
    = 
    \tau_\rho(M),
\end{align}
\end{subequations}
which, comparing with \eqref{e:detddef}, means (cf. \cite[Lemma 5.5]{cattaneo2017cellular})
$$
{\sdet}^\flat(d_\nabla)\equiv{\sdet}^\flat(d_{\nabla}^*) := \tau_\rho(M).
$$
{
Observe that, as mentioned in Remark \ref{Rem:omitphase}, our definition of partition function discards a potential phase factor in Equations~\eqref{heuristicberezininangauge}, which depends on the signature of the operator. However, as observed in \cite[Proposition 5.7 and Remarks 5.8, 5.9]{cattaneo2017cellular}, in the case of acyclic complexes the phase drops out.} The analysis of the partition function of $BF$ theory in the metric gauge fixing is discussed in detail in \cite{cattaneo2017cellular}, where $BF$ theory is cast on triangulated manifolds with boundary (cellular decompositions of cobordisms) and its partition function is shown to coincide with the Reidermeister torsion. The triangulation of a manifold can also be seen as a way of regularising the infinite dimensional BV Laplacian $\Delta_{BV}$, of which $\mathbb{S}_{BF}$ is a cocycle. 
\end{remark}

\begin{remark}
In Schwarz's formalism the computation of the partition function is done by 
performing integration over $\mathrm{ker}(\mathbb{T})\subset \bigoplus_{j=0}^{N-2}V_j$ (see Remark \ref{rem:Schwarzprocedure}). It is tantamount to the $0$-th Koszul--Tate cohomology, i.e. the resolution presented in Equation \eqref{KTres}. This can be compared with the metric gauge in the BV interpretation, where $\mathcal{F}_{BF}\simeq T^*[-1]\bigoplus_{k=0}^{N-2}V_k[k]$. In order to represent integration over $\mathrm{ker}(\mathbb{T})$ within the BV formalism, one considers $\mathbb{L}_g$, which is essentially the same condition --- i.e. restriction to coclosed forms --- extended to the shifted cotangent bundle.
\end{remark}

\subsection{Contact gauge fixing for $BF$ theory}\label{subsec:contact-gauge fixing} 
This subsection shows that if $BF$ theory is cast on a manifold that admits a contact structure, it is possible to find an alternative gauge fixing condition. Adopting the same point of view on the partition function, we show how one recovers the Ruelle zeta function.

Consider the geometric data\footnote{Again, the importance of $E$ is to ensure acyclicity.} from Subsection~\ref{subsec:geo-contact},
$(M,E, X)$, and consider $BF$ theory cast in the Batalin--Vilkovisky formalism as in Definition \ref{Def:BVBF}. 

\begin{proposition}\label{BFgaugefixingproposition}
The submanifold  
$$\LL_X:=\{(\A,\B)\in \mathcal{F}_{BF}\ |\ \iota_X\B=0;\, \iota_X\A=0\}$$ 
is Lagrangian in $\mathcal{F}_{BF}$.
\end{proposition}

\begin{proof}
We show that the condition $\LL_X: \iota_X\B=\iota_X\A=0$ defines an isotropic submanifold with an isotropic complement. Indeed, in virtue of the splitting in Eq. \eqref{Omegasplitting}, we can decompose  $\Omega^\bullet(M)=\mathrm{ker}(\iota_X) \oplus \mathrm{ker}(\alpha\wedge)$, so that $\A=\varphi + \alpha\wedge \eta$ and $\B=\psi + \alpha\wedge \xi$. We observe that $\eta=\xi=0$ on $\LL_X$ and, defining $\LL_X^\perp\coloneqq\{(\A,\B)\in\mathcal{F}_{BF}\ |\ \alpha\wedge \A = \alpha\wedge\B=0\}$, also that $\psi=\phi=0$ on $\LL_X^\perp$. Then we have (we understand the top-form part of all the integrands)
\begin{align}
    \Omega_{BF}\vert_{\LL_X} & = \int\limits_M \left[\delta\psi\delta(\alpha\wedge\eta) + \delta(\alpha\wedge\xi)\delta\varphi\right]_{\LL_X} \equiv 0\\
    \Omega_{BF}\vert_{\LL_X^\perp} & = \int\limits_M \left[\delta\psi\delta(\alpha\wedge\eta) + \delta(\alpha\wedge\xi)\delta\varphi\right]_{\LL_X^\perp} \equiv 0.
\end{align}
This shows that both $\LL_X$ and $\LL_X^\perp$ are isotropic.
\end{proof}

\begin{definition}\label{def:contact-gauge}
We shall refer to the choice of gauge fixing proposed in Proposition \ref{BFgaugefixingproposition} for $BF$ theory as the  {contact gauge}.
\end{definition}

\begin{theorem}\label{Theorem}
Consider the geometric data $(M,E,X)$. Suppose that the Ruelle zeta function for $X$ has a meromorphic extension which does not vanish at zero. Then, the Ruelle zeta function at zero $\zeta_\rho(0)$ computes the partition function of $BF$ theory in the contact gauge:
\begin{equation}
Z(\mathcal{S}_{BF},\,\mathbb{L}_X) = |\zeta_\rho(0)|^{(-1)^{n+1}}.
\end{equation}
\end{theorem}

\begin{proof}
We want to compute $Z(\mathcal{S}_{BF},\,\mathbb{L}_X)\equiv Z(\mathbb{S}_{BF}\vert_{\LL_X})$. 
Because of the decomposition \eqref{Omegasplitting} we have that $\iota_X\B=0 \iff \B=(-1)^{|{\tau}|+1}\iota_X{\tau}$, for some ${\tau}$, hence:
\begin{align}
    \mathbb{S}_{BF}\vert_{\mathbb{L}_X} 
    &= 
    \int\limits_{M}[\B d_\nabla \A]^{\mathrm{top}}_{\LL_X} 
    = 
    \intl_{M}[(-1)^{|{\tau}|}\iota_X{\tau} d_\nabla \A]^{\mathrm{top}}_{\iota_X\A = 0}    \\
    &= 
    \intl_M [\tau \iota_X d_\nabla\A]_{\iota_X\A=0}^{\mathrm{top}} = \intl_{M}[{\tau} \Lie_X\vert_{\Omega^\bullet_0}\A]
\end{align}
Observe that ${\tau}$ and $\A$ are inhomogeneous forms in $\alpha\wedge \Omega_0^\bullet(M,E)[N-2]$ and $\Omega_0^\bullet(M,E)[1]$, respectively. Then, according to Definition \ref{Def:partitionfunction}, recalling that $\Lie_{X}$ acts on a $1$-shifted graded vector space and referring to Definition \ref{Def:detONE} for ${\sdet}^\flat(\iota_X)$, we have
\begin{equation}
	Z(\mathcal{S}_{BF},\,\mathbb{L}_X)
	=
	\left|\sdet^\flat \Lie_{X}\vert_{\Omega^\bullet_0}\right|
	=
	\left|\prod_{k=0}^{N-1} \left({\det}^\flat \Lie_{X,k}\vert_{\Omega^k_0}\right)^{(-1)^k}\right|.
\end{equation}
This, compared with Equation \eqref{Ruelleasadeterminant} and \eqref{Zetadecomposition}, allows us to conclude the proof.
\end{proof}

\begin{remark}
We would like to stress how the crucial element in both metric and contact gauges is the existence of a ``Hodge decomposition'' for the space of $k$-forms. In the metric case it is given, in particular, by the kernel of the de Rham differential and its dual, but truly it can be considered independently of it, as in the contact case, where the maps $\iota_X$ and $\alpha\wedge$ define the splitting.
\end{remark}

\begin{remark}
It is tempting to consider a field theory with action functional $S_X:=\int_{M} [\phi \Lie_X \psi]^{\text{top}}$ on differential forms $\phi,\psi\in \Omega^\bullet(M)$, and call it  {Ruelle theory}. Observe that this is what one gets out of Theorem \ref{Theorem}, but with a nontrivial shift in degree. This is akin to considering the theory $\int_{M}[\mathcal{B}d\mathcal{A}]^{\text{top}}$, for  {unshifted} inhomogeneous differential forms $\mathcal{A},\mathcal{B}\in\Omega^\bullet$, but - to the best of our knowledge - that does not seem to have a clear interpretation up to now.
\end{remark}

\section{Lagrangian homotopies, Fried's conjecture and  gauge-fixing independence}\label{s:homotopies}
In this section we outline a strategy to interpret the recent results of \cite{dang2018fried} on Fried's conjecture as invariance of gauge-fixing in the BV formalism, and vice-versa. We will setup a general geometric framework to discuss perturbation of the Anosov vector field as an argument for the Ruelle zeta function, and argue how this can be used to construct homotopies for Lagrangian submanifolds. Indeed, consider the following:
\begin{claim}
Assume Theorem \ref{BVtheorem} holds on $\mathcal{F}_{BF}$ for some appropriate choice of BV Laplacian $\Delta_{BV}$, and assume there exists a Lagrangian homotopy between $\mathbb{L}_X$ and $\mathbb{L}_g$ for any $X$ contact and Anosov. Then, up to phase, 
\begin{equation}
    \tau_\rho(M) = Z(\mathcal{S}_{BF},\,\mathbb{L}_g) = Z(\mathcal{S}_{BF},\,\mathbb{L}_X) = |\zeta_\rho(0)|^{(-1)^{n}},
\end{equation}
which is the statement of (Fried's) Conjecture \ref{conj:fried}. 
\end{claim}
{
\begin{remark}
We stress here that a universal generalisation of Theorem \ref{BVtheorem} to infinite dimensions, i.e. one that works independently of the details of the field theory (and gauge fixing), is currently not yet available. The results of \cite{cattaneo2017cellular} suggest that one possible regularisation scheme for the BV Laplacian might arise as a limit of finite-dimensional cellular analogues. The works of \cite{costello2011renormalization} and \cite{Gwilliam2012FactorizationAA,GwillRune} assume instead a somewhat different setup, which would also need to be adapted to the case at hand. Ideally, a scheme should be found that makes the regularisation of the BV Laplacian compatible with the context of Definition \ref{flatdet}. Observe that the recent results of \cite{chaubet2019dynamical} suggest that a number of subtleties might arise, already in acyclic cases.
\end{remark} 
}
The central notions for this section, and the main theorem we will need are as follows.

\begin{definition}
Let $X$ be an Anosov vector field on a manifold $M$. We will say that $X$ is  {regular} whenever the restricted resolvents $R_k(\lambda)=(\lambda_{X,k} + \lambda)^{-1}$ of Equation \eqref{eqn:restricted-resolvent} do not have poles at $\lambda=0$ for all $0\le k \le N-1$.
\end{definition}

Recalling Remark \ref{rem:zetafunX}, we report now a result by Dang, Guillarmou, Rivi\`ere and Shen on the properties of Ruelle zeta function seen as a function on Anosov vector fields:
\begin{theorem}[\cite{dang2018fried}]\label{bigtheorem}
Let $(M,E)$ denote the geometric data of Subsection~\ref{subsec:geo-vanilla}.
Consider the set
$U\subset C^\infty(M,TM)$
of regular smooth Anosov vector fields $X$.
Then this set is open, and the map $\zeta\colon U \to \mathbb{C}$, sending $X$ to
$\zeta(X,0)$ is locally constant and nonzero.
\end{theorem}

\subsection{A sphere bundle construction}

Let $\Sigma$ be a compact manifold. We denote by $\mathcal{R}(\Sigma)$ the space of all Riemannian metrics on $\Sigma$, by $\mathcal{R}_{<}(\Sigma)$ the space of negative sectional curvature metrics, and by $\mathcal{R}_h(\Sigma)$ the space of hyperbolic metrics. 

\begin{definition}
Let us fix a reference metric $g_0$. We denote the sphere bundle associated to $g_0$ by $S_0^*\Sigma\coloneqq S_{g_0}^*\Sigma$. Moreover, let $g,\widetilde{g}\in \mathcal{R}(\Sigma)$ and consider the diffeomorphism
\begin{equation}
    \sigma_g^{\widetilde{g}}\colon S_g^*\Sigma \longrightarrow S_{\widetilde{g}}^*\Sigma
\end{equation}
obtained by rescaling lengths in $T\Sigma$. Then, we denote the diffeomorphisms $\sigma_g\colon S_g^*\Sigma \to S_0^*\Sigma$, with $\sigma_g\coloneqq\sigma^{g_0}_g$ for all $g\in\mathcal{R}(\Sigma)$. Denoting by $\varphi_g$ the geodesic flow on $S_g^*\Sigma$, and by $X_g$ the associated vector field. One can transfer the geodesic flow $\varphi_g$ to $S_0^*\Sigma$ by $\varphi_g^0=\sigma_g \circ\varphi_g \circ\sigma_g^{-1}$; we denote the associated vector field on $S_0^*\Sigma$ by $X_g^0$.

\end{definition}

Since we are interested in computing the Ruelle zeta function on sphere bundles, we would like to ensure that the geodesic (Reeb) vector field $X_g$ is Anosov in $S_g^*\Sigma$. This won't be true in general, but it will be true for metrics of negative sectional curvature $\mathcal{R}_<$. The map $\sigma_g$ allows us to consider all Anosov-geodesic vector fields on the same reference space $S_0^*\Sigma$.

\begin{lemma}\label{Anosovpreservation}
The Anosov property of is preserved under pushforward by $\sigma_g$.
\end{lemma}

\begin{proof}
This is an adaptation of a result proved in 
\cite{matsumoto2013space}. 
Given $S_g^*\Sigma$ with Anosov flow $\varphi_g$ and associated vector field $X_g$, we also have a decomposition
\begin{align}
    T(S_g^*\Sigma) = \R X_g + E_s(X_g) + E_u(X_g)
\end{align}
with constants $C_g, \lambda_g$ so that for all $v_s\in E_s(X_g)$, $v_u\in E_u(X_g)$ and all $t>0$:
\begin{align}
    \|d\varphi_{g,t} v_s \|_g \le C_g e^{-\lambda_g t} \|v_s \|_g,
    \qquad
    \|d\varphi_{g,-t} v_u \|_g \le C_g e^{-\lambda_g t} \|v_u \|_g.
\end{align}
Using $d\sigma_g$ we push the decomposition of $T(S_g^*\Sigma)$ to $T(S_0^*\Sigma)$ providing
\begin{align}
    T(S_0^*\Sigma) = \R X_g^0 + E_s(X_g^0) + E_u(X_g^0)
\end{align}
where $E_s(X_g^0) \coloneqq d\sigma_g (E_s(X_g))$ and $E_u(X_g^0) \coloneqq d\sigma_g (E_u(X_g))$. Now $d\sigma_g$ and $d\sigma_g^{-1}$ are both bounded so for $v_s^0\in E_s(X_g^0)$, there is $v_s \in E_s$ such that $d\sigma_g v_s = v_s^0$, whence for $t>0$:
\begin{align}
    \|d\varphi_{g,t}^0 v_s^0\|_0
    &=\|d\sigma_g d\varphi_{g,t} v_s \|_0
    \le \|d\sigma_g\| \|d\varphi_{g,t} v_s \|_g 
    \le \|d\sigma_g\| C_g e^{-\lambda_g t} \|v_s\|_g
    \\
    &= \|d\sigma_g\| C_g e^{-\lambda_g t} \|d \sigma_g^{-1} v_s^0\|_g
    \le \|d\sigma_g\| \|d\sigma_g^{-1}\| C_g e^{-\lambda_g t} \| v_s^0\|_0
    = C_g' e^{-\lambda_g t} \| v_s^0\|_0
\end{align}
A similar result holds for vectors in $E_u(X_g^0)$.
\end{proof}

\begin{definition}
We define the assignment:
\begin{equation}
    \mathbb{X}
    \colon 
    \mathcal{R}(\Sigma)  
    \longrightarrow  
    C^\infty(S_0^*\Sigma; T(S_0^*\Sigma));
    \qquad 
    \mathbb{X}(g)=d\sigma_g X_g=X_g^0
\end{equation}
where $X_g$ is the geodesic vector field induced by $g$ on $S_g^*\Sigma$.
\end{definition}

\begin{proposition}Let us denote by $\mathcal{A}(\Sigma)$ the set of Anosov vector fields on $S_0^*\Sigma$, and by $\mathcal{R}_<(\Sigma)$ the space of negative sectional curvature metrics on $\Sigma$. Then
\begin{equation}
    \mathcal{A}_<(\Sigma) \coloneqq\mathrm{Im}\left(\mathbb{X}\vert_{\mathcal{R}_<(\Sigma)}\right) \subset \mathcal{A}(\Sigma).
\end{equation}
\end{proposition}

\begin{proof}
This follows from the general fact that all negative sectional curvature metrics have Anosov geodesic flows, and application of Lemma \ref{Anosovpreservation}.
\end{proof}

\begin{lemma}\label{surfaceanosovcontractibility}
For $\Sigma$ a 2-dimensional surface of negative Euler characteristic $\chi(\Sigma)<0$, the image of $\mathbb{X}$ restricted to $\mathcal{R}_h(\Sigma)$, the space of hyperbolic metrics on $\Sigma$, is a deformation retract of $\mathcal{A}_{<}(\Sigma)$.
\end{lemma}

\begin{proof}
On $2$-dimensional surfaces, the space of negative sectional curvature metrics coincides with the space of negative-Ricci-curvature metrics, which is a contractible subset of $\mathcal{R}$. In fact,  {via} the Ricci flow one canonically deforms any metric $g$ into a hyperbolic one and, in particular, negatively curved metrics are deformed in such a way along a path of negatively curved metrics 
\cite{tuschmann2015moduli}. 
Thus, the space of hyperbolic metrics on $\Sigma$ is a deformation retract of the space of all negatively curved metrics $\mathcal{R}_<(\Sigma)$, and so will be its image under $\mathbb{X}$ with respect to the space $\mathcal{A}_{<}(\Sigma)$.
\end{proof}

\begin{theorem}\label{thm:Friedconjecture2d}
Consider the contact, Anosov-Riemannian structure {$(M, E, X_g, g)$} of Subsection \ref{subsec:geo-example}, where {$M=S^*_g\Sigma$ with} $\mathrm{dim}(\Sigma)=2$, the Euler characteristic $\chi(\Sigma)<0$ and such that $g\in\mathcal{R}_<(\Sigma)$. Then,
{

\begin{equation}
    |\zeta(X_g,0)|=\tau_\rho(M)^{-1}.
\end{equation}
}
\end{theorem}

\begin{proof}
In virtue of Lemma \ref{surfaceanosovcontractibility} there exists a smooth path $X(t)$ of Anosov vector fields in $\mathcal{A}_<(\Sigma)$ connecting $X_g$ to $X_{g_h}$ for some $g_h$ hyperbolic, where Theorem \ref{thm:FCT} holds: let $g(t)$ be a retraction such that $g(0)=g_0$ and $g(1)=g$, and set $X(t)=\mathbb{X}(g(t))$. If we can show that each $X(t)$ is regular then local constancy of the zeta function at zero (Theorem \ref{bigtheorem}) provides the result.

The idea for showing regularity is essentially in \cite{dyatlov2017ruelle} even though that setting is using the trivial representation, so that the de Rham complex is not acyclic. A more explicit version which adapts the spirit of \cite{dyatlov2017ruelle} is in \cite{dang2018fried}. Below we cite the necessary results of \cite{dang2018fried} to conclude that $X(t)$ is regular.

Since $g\in\mathcal{R}_<(\Sigma)$, its associated geodesic vector field $X_g$ on $S_g^*\Sigma$ is Anosov \cite{anosov1967geodesic, anosov1967some, arnold1968problemes}. 
For regularity, one shows that no $k$-form resonant states $\varphi^{(k)}$ associated with the spectral parameter $\lambda=0$ exist for $k\in\{0,1,2\}$ 
(recall Equation~\eqref{eqn:resonant-states-characterisation}).
By \cite[Lemma 7.4]{dang2018fried}, 0-form resonant states $\varphi^{(0)}$ are closed and smooth in this setting so may be identified with degree-0 de Rham cohomology. Acyclicity implies such states are necessarily the 0-section. For 2-form resonant states $\varphi^{(k)}$, there is an isomorphism with 0-form resonant states upon wedging with $d\alpha$ \cite[Lemma 7.2]{dang2018fried}. 
Effectively 
$\varphi^{(2)}=\varphi^{(0)}d\alpha$ 
so in this case also, no such resonant states exist. Finally, for 1-form resonant states, we require degree-1 cohomology to vanish and appeal to \cite[Lemma 7.2 Hypothesis (1)]{dang2018fried} to conclude no non-trivial resonant states exist.
\end{proof}

\begin{corollary}
Under the assumptions of Theorem \ref{thm:Friedconjecture2d}, denoting by
$$\mathbb{L}_X=\{(\A,\B)\in\mathcal{F}_{BF}\ |\ \iota_{\mathbb{X}(g)}\A=\iota_{\mathbb{X}(g)}\B=0\}$$
the Lagrangian submanifold defined by the Anosov vector field $\mathbb{X}(g)$, we have
{\begin{equation}
    Z(\mathcal{S}_{BF},\,\mathbb{L}_X)=\tau_\rho(M)^{-1},
\end{equation}}
for every $g\in\mathcal{R}_<(\Sigma)$.
\end{corollary}
\begin{proof}
This follows from Theorems \ref{Theorem} and \ref{thm:Friedconjecture2d}.
\end{proof}

\begin{remark}
Observe that the above construction can be generalised to some extent. In general, the space of negative sectional curvature metric will not be path connected, so let us consider the connected components that contain a hyperbolic metric. The image under $\mathbb{X}$ of the disjoint union of such connected components in $\mathcal{A}(\Sigma)$ will be a union of  {islands} of Anosov vector fields, path connected to an anosov vector fields in $\mathrm{Im}(\mathbb{X}\vert_{\mathcal{R}_h})$. In a neighborhood of a hyperbolic metric Anosov vector field, the requirements of Theorem \ref{bigtheorem} are satisfied, and Fried's conjecture might be extended to open subsets of $\mathcal{A}_<(\Sigma)$. We defer the development of such a generalisation to a subsequent work. 
\end{remark}

\begin{remark}
The spirit of gauge-fixing homotopies is that of replacing an ill-defined integral with a well-defined one, and can be considered as providing a family of integral representations of the same quantity, only some of which are directly computable. From this point of view, the Ruelle zeta function at zero might not be computable for a generic (Anosov) vector field, but it will be once deformed away from an invalid point in $\mathcal{A}(\Sigma)$.
\end{remark}

We wish to interpret this result in terms of homotopies of Lagrangian submanifolds in $\mathcal{F}_{BF}$ and gauge fixing independence for $BF$ theory.

\begin{theorem}
Consider the geometric data {$(M, E, X_g, g)$} of Subsection \ref{subsec:geo-example}, a smooth path $ g_t: [0,1] \to \mathcal{R}(\Sigma)$ such that $\mathbb{X}(g_0)$ is regular, and let 
$$
\LL_t:=\{(\A,\B)\in\mathcal{F}_{BF}\ |\ \iota_{\mathbb{X}(g_t)}\A=\iota_{\mathbb{X}(g_t)}\B=0\}
$$ 
be the associated smooth family of Lagrangian submanifolds in $\mathcal{F}_{BF}$. Then,  $Z(\mathcal{S}_{BF},\,\mathbb{L}_0)\not = 0$ and
\begin{equation}
    \frac{d}{dt}\Big\vert_{t=0} Z(\mathcal{S}_{BF},\,\mathbb{L}_t) =0.
\end{equation}
Moreover, if $g_0$ is hyperbolic, we have that
\begin{equation}
    Z(\mathcal{S}_{BF},\,\mathbb{L}_0) = \tau_\rho(M).
\end{equation}
\end{theorem}

\begin{proof}
Lemma \ref{Anosovpreservation} ensures that $\mathbb{X}(g_t)=d\sigma_{g_t}X_{g_t}$ is a smooth path of Anosov vector fields in $S_0^*\Sigma$, and in virtue of Theorem \ref{Theorem} we have that 
$$
Z(\mathcal{S}_{BF},\,\mathbb{L}_t) \equiv Z(\mathbb{S}_R\vert_{\mathbb{L}_t}) = |\zeta(\mathbb{X}(g_t),0)|^{(-1)^n}.
$$ 
By assumption, at $g_0$ the $k$-th Ruelle zeta factors $\zeta_k(\mathbb{X}(g_0),0)$ are well defined and different from zero. Then, in virtue of Theorem \ref{bigtheorem} $\zeta(\mathbb{X}(g_0),0)$  is constant in a open neighborhood of $\mathbb{X}(g_0)$, hence it is on the whole path $g_t$ for $t\in[0,T)$ for some appropriately chosen $T$ and in particular its derivative at $t$ vanishes. If we choose $g_0$ hyperbolic, using Theorem \ref{thm:FCT} we can conclude that 
$$
Z(\mathcal{S}_{BF},\,\mathbb{L}_0) = |\zeta(\mathbb{X}(g_0),0)|^{(-1)^n} \equiv |\zeta_{\rho}(M)|^{(-1)^n}= \tau_\rho(M).
$$
\end{proof}

\section*{Acknowledgements} This research was (partly) supported by the NCCR SwissMAP, funded by the Swiss National Science Foundation. MS acknowledges partial support from Swiss National Science Foundation grants P2ZHP2\_164999 and P300P2\_177862. SK acknowledges partial support from Deutsche Forschungs Gemeinschaft (DFG) under Graduiertenkolleg 1821 (Cohomological Methods in Geometry). This project began, and was mostly developed, during the period when the first and last author were at UC Berkeley, to which we are grateful for having provided such a stimulating research environment. M.S. would like to thank the University of Freiburg for hospitality. We are indebted to A.S. Cattaneo and P. Mnev for their support and feedback at different stages of this work. 
Since the preprint appeared on arXiv, we have received several useful comments from Y. Chaubet and N.V. Dang whom we warmly thank. We also thank the anonymous referee for providing useful comments that helped improve this paper.

\bibliographystyle{alpha}
\bibliography{Ruelle-BF-Arxiv.bib}

\begin{thebibliography}{CCRFM95}

\bibitem[AA68]{arnold1968problemes}
V.I. Arnold and A.~Avez.
\newblock Probl{\`e}mes ergodiques de la m{\'e}canique classique, {P}aris,
  {G}authier-{V}illars (1967); ({E}nglish translation: {E}rgodic problems in
  classical mechanics) {N}ew {Y}ork, 1968.

\bibitem[AB64]{atiyah1964notes}
M.F. Atiyah and R.~Bott.
\newblock {\em Notes on the Lefschetz fixed point theorem for elliptic
  complexes}.
\newblock Harvard University, 1964.

\bibitem[AB68]{atiyah1968lefschetz}
M.F. Atiyah and R.~Bott.
\newblock A {L}efschetz fixed point formula for elliptic complexes: {II}.
  {A}pplications.
\newblock {\em Annals of Mathematics}, pages 451--491, 1968.

\bibitem[And92]{anderson1992introduction}
I.M. Anderson.
\newblock Introduction to the variational bicomplex.
\newblock {\em Contemporary Mathematics}, 132, 1992.

\bibitem[Ano67]{anosov1967geodesic}
D.V. Anosov.
\newblock Geodesic flows on closed {R}iemannian manifolds of negative
  curvature.
\newblock {\em Trudy Mat. Inst. Steklov}, 90(5), 1967.

\bibitem[AS67]{anosov1967some}
D.V. Anosov and Y.G. Sinai.
\newblock Some smooth ergodic systems.
\newblock {\em Russian Mathematical Surveys}, 22(5):103, 1967.

\bibitem[Bal18]{baladi2018dynamical}
V.~Baladi.
\newblock {\em Dynamical zeta functions and dynamical determinants for
  hyperbolic maps}.
\newblock Springer, 2018.

\bibitem[BBRT91]{BBRT91}
D.~Birmingham, M.~Blau, M.~Rakowski, and G.~Thompson.
\newblock Topological field theory.
\newblock {\em Physics Reports}, 209(4):129 -- 340, 1991.

\bibitem[BE15]{Berwick2015}
D.~Berwick-Evans.
\newblock The {C}hern--{G}auss--{B}onnet theorem via supersymmetric {E}uclidean
  field theories.
\newblock {\em Communications in Mathematical Physics}, 335(3):1121--1157, May
  2015.

\bibitem[Ber83]{berezin1983introduction}
F.A. Berezin.
\newblock Introduction to algebra and analysis with anticommuting variables.
\newblock {\em Moscow Univ}, 1983.

\bibitem[BL75]{berezin1975supermanifolds}
F.A. Berezin and D.A. Leites.
\newblock Supermanifolds.
\newblock In {\em Doklady Akademii Nauk}, volume 224, pages 505--508. Russian
  Academy of Sciences, 1975.

\bibitem[BL07]{butterley2007smooth}
O.~Butterley and C.~Liverani.
\newblock Smooth {A}nosov flows: correlation spectra and stability.
\newblock {\em J. Mod. Dyn}, 1(2):301--322, 2007.

\bibitem[BRS74]{becchi1974abelian}
C.~Becchi, A.~Rouet, and R.~Stora.
\newblock The abelian higgs kibble model, unitarity of the s-operator.
\newblock {\em Physics Letters B}, 52(3):344--346, 1974.

\bibitem[BRS75]{becchi1975renormalization}
C.~Becchi, A.~Rouet, and R.~Stora.
\newblock Renormalization of the abelian {H}iggs-{K}ibble model.
\newblock {\em Communications in Mathematical Physics}, 42(2):127--162, 1975.

\bibitem[BRS76]{becchi1976renormalization}
C.~Becchi, A.~Rouet, and R.~Stora.
\newblock Renormalization of gauge theories.
\newblock {\em Annals of Physics}, 98(2):287--321, 1976.

\bibitem[BV83]{batalin1983quantization}
I.A. Batalin and G.A. Vilkovisky.
\newblock Quantization of gauge theories with linearly dependent generators.
\newblock {\em Physical Review D}, 28(10):2567, 1983.

\bibitem[BV84]{batalin1984gauge}
I.A. Batalin and G.A. Vilkovisky.
\newblock Gauge algebra and quantization.
\newblock In {\em Quantum Gravity}, pages 463--480. Springer, 1984.

\bibitem[CC18]{cattaneo2018split}
A.S. Cattaneo and I.~Contreras.
\newblock Split canonical relations.
\newblock {\em arXiv preprint arXiv:1811.10107}, 2018.

\bibitem[CCRFM95]{cattaneo1995topological}
A.S. Cattaneo, P.~Cotta-Ramusino, J.~Fr{\"o}hlich, and M.~Martellini.
\newblock Topological {BF} theories in 3 and 4 dimensions.
\newblock {\em Journal of Mathematical Physics}, 36(11):6137--6160, 1995.

\bibitem[CD19]{chaubet2019dynamical}
Y.~Chaubet and N.V. Dang.
\newblock Dynamical torsion for contact {A}nosov flows.
\newblock {\em arXiv preprint arXiv:1911.09931}, 2019.

\bibitem[CFL06]{CFL}
A.S. Cattaneo, D.~Fiorenza, and R.~Longoni.
\newblock Graded {P}oisson algebras.
\newblock In J.P. Fran{\c c}oise, G.L. Naber, and T.S. Tsun, editors, {\em
  Encyclopedia of Mathematical Physics}, pages 560 -- 567. Academic Press,
  Oxford, 2006.

\bibitem[CG16]{costello2016factorization}
K.~Costello and O.~Gwilliam.
\newblock {\em Factorization algebras in quantum field theory}, volume~1.
\newblock Cambridge University Press, 2016.

\bibitem[CMR95]{cordes1995lectures}
S.~Cordes, G.~Moore, and S.~Ramgoolam.
\newblock Lectures on 2d {Y}ang-{M}ills theory, equivariant cohomology and
  topological field theories.
\newblock {\em Nuclear Physics B - Proceedings Supplements}, 41(1):184 -- 244,
  1995.

\bibitem[CMR14]{cattaneo2012classical}
A.~S. Cattaneo, P.~Mnev, and N.~Reshetikhin.
\newblock Classical {BV} theories on manifolds with boundary.
\newblock {\em Communications in Mathematical Physics}, 332(2):535--603, 2014.

\bibitem[CMR18]{cattaneo2018perturbative}
A.S. Cattaneo, P.~Mnev, and N.~Reshetikhin.
\newblock Perturbative quantum gauge theories on manifolds with boundary.
\newblock {\em Communications in Mathematical Physics}, 357(2):631--730, 2018.

\bibitem[CMR20]{cattaneo2017cellular}
Alberto~S. Cattaneo, Pavel Mnev, and Nicolai Reshetikhin.
\newblock A cellular topological field theory.
\newblock {\em Communications in Mathematical Physics}, 374(2):1229--1320,
  2020.

\bibitem[Cos11]{costello2011renormalization}
K.~Costello.
\newblock {\em Renormalization and effective field theory}, volume 170.
\newblock American Mathematical Soc., 2011.

\bibitem[CR01]{cattaneo2001higher}
A.S. Cattaneo and C.A. Rossi.
\newblock Higher-dimensional {BF} theories in the {B}atalin--{V}ilkovisky
  formalism: The {BV} action and generalized {W}ilson loops.
\newblock {\em Communications in Mathematical physics}, 221(3):591--657, 2001.

\bibitem[CS11]{cattaneo2011introduction}
A.S. Cattaneo and F.~Sch{\"a}tz.
\newblock Introduction to supergeometry.
\newblock {\em Reviews in Mathematical Physics}, 23(06):669--690, 2011.

\bibitem[Del17]{delgado2017lagrangian}
N.L. Delgado.
\newblock {\em Lagrangian field theories: ind/pro-approach and L-infinity
  algebra of local observables}.
\newblock PhD thesis, Max Planck Institute for Mathematics, 2017.

\bibitem[DF99]{deligne1999classical}
P.~Deligne and D.S. Freed.
\newblock Classical field theory.
\newblock {\em Quantum fields and strings: a course for mathematicians},
  2:137--226, 1999.

\bibitem[DGRS20]{dang2018fried}
Nguyen~Viet Dang, Colin Guillarmou, Gabriel Rivi{\`e}re, and Shu Shen.
\newblock The {F}ried conjecture in small dimensions.
\newblock {\em Inventiones mathematicae}, 220(2):525--579, 2020.

\bibitem[DZ16]{dyatlov2016dynamical}
S.~Dyatlov and M.~Zworski.
\newblock {Dynamical zeta functions for {A}nosov flows via microlocal
  analysis}.
\newblock {\em Ann. Sci. {\'{E}}c. Norm. Sup{\'{e}}r. (4)}, 49(3):543--577,
  2016.

\bibitem[DZ17]{dyatlov2017ruelle}
S.~Dyatlov and M.~Zworski.
\newblock Ruelle zeta function at zero for surfaces.
\newblock {\em Inventiones mathematicae}, 210(1):211--229, 2017.

\bibitem[FP16]{faddeev2016feynman}
L.D. Faddeev and V.N. Popov.
\newblock Feynman diagrams for the {Y}ang-{M}ills field.
\newblock In {\em Fifty Years of Mathematical Physics: Selected Works of Ludwig
  Faddeev}, pages 157--158. World Scientific, 2016.

\bibitem[FR13]{fredenhagen2013batalin}
K.~Fredenhagen and K.~Rejzner.
\newblock {B}atalin-{V}ilkovisky formalism in perturbative algebraic quantum
  field theory.
\newblock {\em Communications in Mathematical Physics}, 317(3):697--725, 2013.

\bibitem[Fri86]{fried1986analytic}
D.~Fried.
\newblock Analytic torsion and closed geodesics on hyperbolic manifolds.
\newblock {\em Inventiones mathematicae}, 84(3):523--540, 1986.

\bibitem[Fri87]{fried1987lefschetz}
D.~Fried.
\newblock Lefschetz formulas for flows.
\newblock {\em Contemporary Mathematics}, 58:19--69, 1987.

\bibitem[Fri95]{fried1995meromorphic}
D.~Fried.
\newblock Meromorphic zeta functions for analytic flows.
\newblock {\em Communications in {M}athematical {P}hysics}, 174(1):161--190,
  1995.

\bibitem[GH18]{GwillRune}
Owen Gwilliam and Rune Haugseng.
\newblock Linear {B}atalin--{V}ilkovisky quantization as a functor of
  {$\infty$}-categories.
\newblock {\em Selecta Mathematica}, 24(2):1247--1313, 2018.

\bibitem[GLP13]{giulietti2013anosov}
P.~Giulietti, C.~Liverani, and M.~Pollicott.
\newblock {A}nosov flows and dynamical zeta functions.
\newblock {\em Annals of Mathematics}, pages 687--773, 2013.

\bibitem[Gui77]{guillemin1977lectures}
V.~Guillemin.
\newblock Lectures on spectral theory of elliptic operators.
\newblock {\em Duke Mathematical Journal}, 44(3):485--517, 1977.

\bibitem[Gut91]{gutzwiller1991chaos}
M.C. Gutzwiller.
\newblock {\em Chaos in Classical and Quantum Mechanics}, volume~1.
\newblock Springer Science \& Business Media, 1991.

\bibitem[Gwi]{Gwilliam2012FactorizationAA}
Owen Gwilliam.
\newblock {\em Factorization Algebras and Free Field Theories}.
\newblock PhD thesis, Northwestern {U}niveristy.

\bibitem[Had18]{hadfield2018zeta}
C.S. Hadfield.
\newblock Zeta function at zero for surfaces with boundary.
\newblock {\em arXiv preprint arXiv:1803.10982}, 2018.

\bibitem[Hen90]{henneaux1990lectures}
M.~Henneaux.
\newblock Lectures on the antifield-{BRST} formalism for gauge theories.
\newblock {\em Nuclear Physics B-Proceedings Supplements}, 18(1):47--105, 1990.

\bibitem[KW20]{kuster2019pollicott}
Benjamin K{\"u}ster and Tobias Weich.
\newblock Pollicott-ruelle resonant states and betti numbers.
\newblock {\em Communications in Mathematical Physics}, 378(2):917--941, 2020.

\bibitem[Mar11]{marklof2004selberg}
J.~Marklof.
\newblock Selberg's trace formula: an introduction.
\newblock In J.~Bolte and F.~Steiner, editors, {\em Hyperbolic Geometry and
  Applications in Quantum Chaos and Cosmology}, pages 83--119. Cambridge
  University Press, 2011.

\bibitem[Mat13]{matsumoto2013space}
S.~Matsumoto.
\newblock The space of (contact) {A}nosov flows on 3-manifolds.
\newblock {\em J. Math. Sci. Univ. Tokyo}, 20:445--460, 2013.

\bibitem[Mne14]{mnev2014lecture}
P.~Mnev.
\newblock Lecture notes on torsions.
\newblock {\em arXiv preprint arXiv:1406.3705}, 2014.

\bibitem[Mne19]{mnev2019quantum}
P.~Mnev.
\newblock {\em Quantum Field Theory: Batalin--Vilkovisky Formalism and Its
  Applications}, volume~72.
\newblock American Mathematical Soc., 2019.

\bibitem[MS91]{moscovici1991r}
H.~Moscovici and R.J. Stanton.
\newblock R-torsion and zeta functions for locally symmetric manifolds.
\newblock {\em Inventiones mathematicae}, 105(1):185--216, 1991.

\bibitem[RS71]{ray1971r}
D.B. Ray and I.M. Singer.
\newblock R-torsion and the {L}aplacian on riemannian manifolds.
\newblock {\em Advances in Mathematics}, 7(2):145--210, 1971.

\bibitem[Rue76]{ruelle1976zeta}
D.~Ruelle.
\newblock Zeta-functions for expanding maps and {A}nosov flows.
\newblock {\em Inventiones mathematicae}, 34(3):231--242, 1976.

\bibitem[Rue86]{ruelle1986resonances}
D.~Ruelle.
\newblock Resonances of chaotic dynamical systems.
\newblock {\em Physical review letters}, 56(5):405, 1986.

\bibitem[Sch78]{schwarz1978partition}
A.S. Schwarz.
\newblock The partition function of degenerate quadratic functional and
  {R}ay-{S}inger invariants.
\newblock {\em Letters in Mathematical Physics}, 2(3):247--252, 1978.

\bibitem[Sch79]{schwarz1979partition}
A.S. Schwarz.
\newblock The partition function of a degenerate functional.
\newblock {\em Communications in Mathematical Physics}, 67(1):1--16, 1979.

\bibitem[She17]{shen2017analytic}
S.~Shen.
\newblock Analytic torsion, dynamical zeta functions, and the {F}ried
  conjecture.
\newblock {\em Analysis \& PDE}, 11(1):1--74, 2017.

\bibitem[Sta98]{stasheff1998secret}
J.~Stasheff.
\newblock The (secret?) homological algebra of the {B}atalin-{V}ilkovisky
  approach.
\newblock In {\em Conference Secondary Calculus and Cohomological Physics},
  volume 219, pages 195--210. Moscow, August 1998.

\bibitem[TW15]{tuschmann2015moduli}
W.~Tuschmann and D.J. Wraith.
\newblock {\em Moduli spaces of Riemannian metrics}.
\newblock Springer, 2015.

\bibitem[Tyu75]{tyutin1975gauge}
I.V. Tyutin.
\newblock Gauge invariance in field theory and statistical physics in operator
  formalism.
\newblock {\em Lebedeiv Physics Institute preprint arXiv:0812.0580}, 1975.

\bibitem[Vor91]{voronov1991geometric}
T.~Voronov.
\newblock {\em Geometric integration theory on supermanifolds}, volume~1.
\newblock CRC Press, 1991.

\bibitem[Wei71]{weinstein1971symplectic}
A.~Weinstein.
\newblock Symplectic manifolds and their {L}agrangian submanifolds.
\newblock {\em Advances in Mathematics}, 6(3):329--346, 1971.

\bibitem[Wei10]{weinstein2010symplectic}
A.~Weinstein.
\newblock Symplectic categories.
\newblock {\em Portugaliae Mathematica}, 67(2):261--278, 2010.

\bibitem[Zwo17]{zworski2017mathematical}
M.~Zworski.
\newblock Mathematical study of scattering resonances.
\newblock {\em Bulletin of Mathematical Sciences}, 7(1):1--85, 2017.

\end{thebibliography}

\end{document}